\documentclass[journal]{IEEEtran}
\usepackage{amsmath}
\usepackage{amsfonts}
\usepackage{amsthm}

\usepackage[ruled]{algorithm2e}
\usepackage{xcolor}

\usepackage{cite}

\newtheorem*{remark}{Remark}

\newtheorem{theorem}{Theorem}
\newtheorem{lemma}{Lemma}

\usepackage{graphicx}

\newcommand{\traveltime}[1]{\ensuremath{k^{#1}}}

\newcommand{\platooningbenefit}[1]{\ensuremath{b^{#1}}}

\newcommand{\threshold}[2]{\ensuremath{ \rho ^{#1}_{#2}}} 

\newcommand{\releaseset}[2]{\ensuremath{ \mathcal D^{#2}_{#1}}}

\newcommand{\filter}[2]{\ensuremath{\pi^{#1}_{#2}}}
\newcommand{\filterj}[3]{\ensuremath{\pi^{#1}_{#2,#3}}}

\newcommand{\transprob}[4]{\ensuremath{p^{#1}_{#2,#3#4}}}

\newcommand{\obsprob}[5]{\ensuremath{r^{#1}_{#2,#3#4}(#5)}}

\newcommand{\Internalarrival}[2]{\ensuremath{\Theta^{#1}_{#2}}}

\newcommand{\internalarrival}[2]{\ensuremath{\theta^{#1}_{#2}}}

\newcommand{\internalarrivalsample}[3]{\ensuremath{\theta^{#1}_{#2,#3}}}

\ifCLASSOPTIONcompsoc
  \usepackage[caption=false,font=normalsize,labelfont=sf,textfont=sf]{subfig}
\else
  \usepackage[caption=false,font=footnotesize]{subfig}
\fi

\begin{document}

%\title{Decentralized Platoon Coordination in Hub-Corridor: \\ Optimal Release Rules and  Approximations  }

\title{Hub-Based Platoon Formation:  Optimal Release Policies and Approximate Solutions  }

\author{Alexander~Johansson,
       Ehsan~Nekouei,  Xiaotong~Sun, \\ Karl Henrik Johansson,
        and~Jonas~M{\aa}rtensson% <-this % stops a space
\thanks{The work by A. Johansson, K. H. Johansson and J. M\aa rtensson was partially supported by Horizon 2020 through the project ENSEMBLE, the Strategic Vehicle Research
and Innovation Programme, the Knut and Alice Wallenberg Foundation by a Wallenberg Scholar Grant,  and
the Swedish Research Council through the Distinguished Professor Grant 2017-01078.  The work by E.~Nekouei is supported by the start-up grant 7200658 from City University of Hong Kong.}     
	\thanks{A. Johansson, K. H. Johansson and J. M\aa rtensson are with the Integrated Transport Research Lab and Division of Decision and Control,
	School of Electrical Engineering and Computer Science, KTH Royal Institute
	of Technology, Stockholm, Sweden.,
	SE-100 44 Stockholm, Sweden. Emails:
	{\tt\small \{alexjoha, kallej, jonas1\}@kth.se}.  X. Sun is with the Thrust of Intelligent Transportation, the Hong Kong University of Science and Technology (Guangzhou), China., and the Department of Civil and Environmental Engineering, the Hong Kong University of Science and Technology, Hong Kong SAR, China. 
	Email: {\tt \small \{xtsun\}@ust.hk}. E. Nekouei is with the Department of Electrical Engineering,
	City University of Hong Kong, Hong Kong. 
	Email: {\tt \small \{enekouei\}@cityu.edu.hk}. }
\thanks{Manuscript submitted November 09, 2023}}

% note the % following the last \IEEEmembership and also \thanks - 
% these prevent an unwanted space from occurring between the last author name
% and the end of the author line. i.e., if you had this:
% 
% \author{....lastname \thanks{...} \thanks{...} }
%                     ^------------^------------^----Do not want these spaces!
%
% a space would be appended to the last name and could cause every name on that
% line to be shifted left slightly. This is one of those "LaTeX things". For
% instance, "\textbf{A} \textbf{B}" will typeset as "A B" not "AB". To get
% "AB" then you have to do: "\textbf{A}\textbf{B}"
% \thanks is no different in this regard, so shield the last } of each \thanks
% that ends a line with a % and do not let a space in before the next \thanks.
% Spaces after \IEEEmembership other than the last one are OK (and needed) as
% you are supposed to have spaces between the names. For what it is worth,
% this is a minor point as most people would not even notice if the said evil
% space somehow managed to creep in.

% The paper headers
\markboth{}%
{A. Johansson \MakeLowercase{\textit{et al.}}: Bare Demo of IEEEtran.cls for IEEE Journals}
% The only time the second header will appear is for the odd numbered pages
% after the title page when using the twoside option.
% 
% *** Note that you probably will NOT want to include the author's ***
% *** name in the headers of peer review papers.                   ***
% You can use \ifCLASSOPTIONpeerreview for conditional compilation here if
% you desire.

% If you want to put a publisher's ID mark on the page you can do it like
% this:
%\IEEEpubid{0000--0000/00\$00.00~\copyright~2015 IEEE}
% Remember, if you use this you must call \IEEEpubidadjcol in the second
% column for its text to clear the IEEEpubid mark.

% use for special paper notices
%\IEEEspecialpapernotice{(Invited Paper)}

% and under the assumption that hubs does not cooperate in sharing their release decisions with others.
% make the title area
\maketitle

% As a general rule, do not put math, special symbols or citations
% in the abstract or keywords.
\begin{abstract}
This paper studies the optimal hub-based platoon formation at hubs along a highway under decentralized, distributed, and centralized policies. Hubs are locations along highways where trucks can wait for other trucks to form platoons. A coordinator at each hub decides the departure time of trucks, and the released trucks from the hub will form platoons. The problem is cast as an optimization problem where the objective is to maximize the platooning reward. We first show that the optimal release policy in the decentralized case, where the hubs do not exchange information, is to release all trucks at the hub when the number of trucks exceeds a threshold computed by dynamic programming. We develop efficient approximate release policies for the dependent arrival case using this result. To study the value of information exchange among hubs on platoon formation, we next study the distributed and centralized platoon formation policies which require information exchange among hubs. To this end, we develop receding horizon solutions for the distributed and centralized platoon formation at hubs using the dynamic programming technique. Finally, we perform a simulation study over three hubs in northern Sweden. The profits of the decentralized policies are shown to be approximately $3.5\%$ lower than the distributed policy and $8\%$  lower than the centralized release policy. This observation suggests that decentralized policies are prominent solutions for hub-based platooning as they do not require information exchange among hubs and can achieve a similar performance compared with distributed and centralized policies. 

% or when the preceding hub along  the highway has independent arrivals and follows
%an optimal release policy

\end{abstract}

%For the second information structure, we propose a solution where the each of the coordinators periodically receives the release decision from its preceding hub and updates its release decision.
%We study the platoon release time problem under the two following information structures: (1) the local coordinators do not share their release decisions with others, and (2) they cooperate in sharing their release decisions.
% Note that keywords are not normally used for peerreview papers.

% For peer review papers, you can put extra information on the cover
% page as needed:
% \ifCLASSOPTIONpeerreview
% \begin{center} \bfseries EDICS Category: 3-BBND \end{center}
% \fi
%
% For peerreview papers, this IEEEtran command inserts a page break and
% creates the second title. It will be ignored for other modes.
\IEEEpeerreviewmaketitle

\section{Introduction}\label{sec:intro}

\subsection{Motivation}

\IEEEPARstart{I}{n} the truck platooning technology, a set of trucks drive with small inter-vehicular distances. Typically, a human driver maneuvers the lead truck in a platoon, while automated driving systems maneuver the follower trucks. The truck platooning  technology reduces the operational cost of road transportation by lightening drivers' workload. The fuel consumption level of the follower trucks in a platoon reduces due to the small inter-vehicular distances between trucks.  Thus, the truck platooning technology results in significant environmental benefits. For instance, energy savings of approximately $10 \%$  for follower trucks have been reported in the literature based on experimental data \cite{Browand2004, Davila2013, Alam2015, Tsugawa2016, Bishop2017}. Other benefits of truck platooning include increased road capacity and safety and reduced travel time \cite{Ioannou1993, Fernandes2012, Jo2019}. All these potential benefits motivate the investigations of different platooning architectures and business models, \emph{e.g.}, see,  \cite{Janssen2015, Chottani2018, Axelsson2020}.

%old
%\IEEEPARstart{P}{latooning} refers to the technology employed when trucks drive in a string on the road with small inter-vehicular distances. Typically, a human driver maneuvers the lead truck in a platoon while automated driving systems maneuver the follower trucks, which directly reduces the operational costs by lightening drivers' workload. Platooning also brings considerable environmental benefits due to the small inter-vehicular distances between trucks, leading to reduced air drag and reduced energy consumption. These benefits have been demonstrated in previous research efforts; for example, \cite{Browand2004, Davila2013, Alam2015,Tsugawa2016, Bishop2017}  reported energy savings of around $10 \%$ for follower trucks in platoons. External benefits of platooning include increased road capacity and safety, as well as reduced travel time \cite{Ioannou1993, Fernandes2012, Jo2019}. All these benefit potentials motivate the investigations of business models for truck platooning, such as \cite{Janssen2015, Chottani2018, Axelsson2020}. 

\subsection{Related work}

	\begin{figure}
	\centering
	\subfloat[Hub-corridor with $H$ hubs ]{\label{Fig:Line_hubs_a} 	\includegraphics[width=9cm]{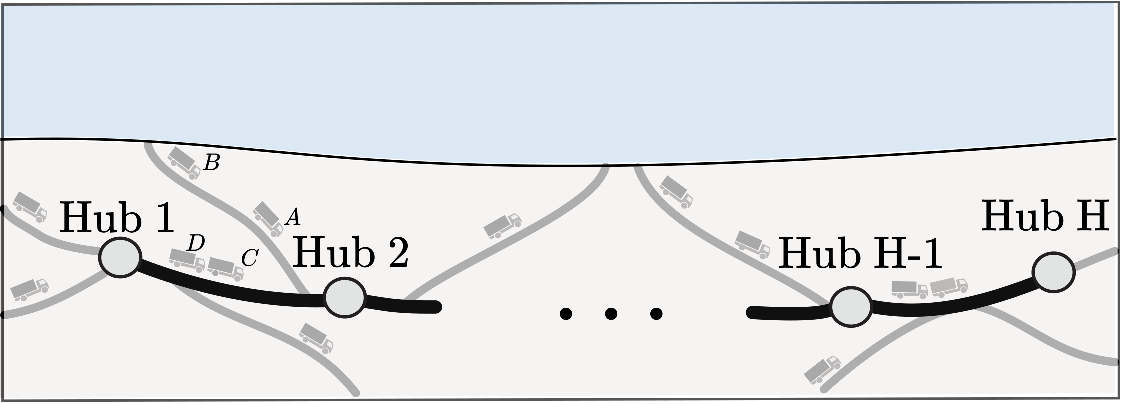}} \\
	
	\subfloat[Section of a hub-corridor including hubs $h-1$ and $h$]{\label{Fig:Line_hubs_b} \includegraphics[width=9cm]{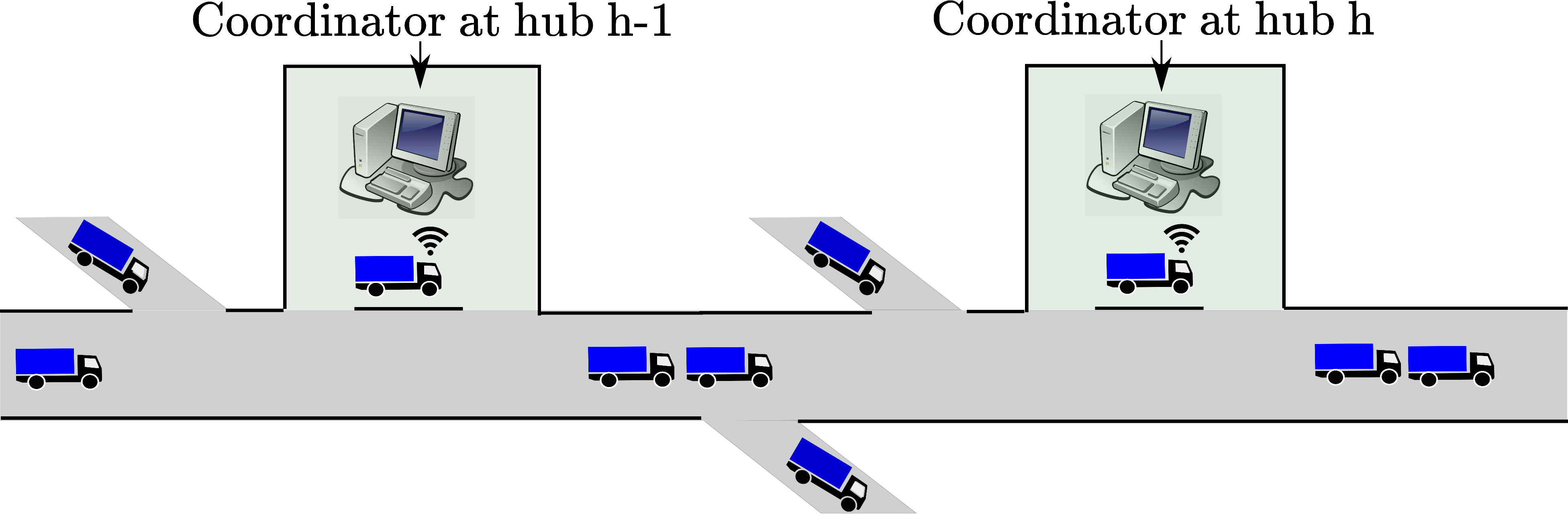}} \\
	
	\caption{ (a) A hub-corridor with $H$ hubs in a transportation network, and (b) the section of a hub-corridor including hubs $h-1$ and $h$ where coordinators make platoon release decisions }
	\label{Fig:Line_hubs}
\end{figure}

Platoon coordination is a central part of the truck platooning technology. Here, platoon coordination refers to the decision-making of which trucks should form a platoon depending on their travel time schedules and their geographical locations. Trucks can merge into platoons on the road by adjusting their speeds. For example see  \cite{Larsson2015,Liang2016,Hoef2018,Xiong2021}. One drawback of on-road platoon formation is their complexity level, as they require trucks to speed up or slow down to merge into platoons on the roads. Such maneuvers might be challenging to perform depending on the surrounding traffic conditions. Another drawback of on-road platoon formation is that the surrounding traffic may be disturbed by the platoon formation or may disturb the platoon formation. An alternative platoon formation approach is hub-based platoon formation. Hubs are locations along highways where trucks can stop and wait for others to form platoons, \emph{e.g.,} freight terminals, gas stations, parking places, tolling stations, and harbors. Under this approach, a platoon coordinator decides on the departure time of trucks from the hubs; the trucks which leave a hub simultaneously form a platoon. Platoon formation at a single hub when the coordinator knows trucks' arrival schedules to the hub was considered in \cite{Zhang2017, Boysen2018, Larsen2019}. In \cite{Johansson2021,   Bai2021}, the authors studied the agent-based platoon formation at multiple hubs in a transportation network, where trucks individually decide their departure times from hubs. In the agent-based approaches, the trucks are required to share their arrival schedules at hubs with each other. In \cite{Johansson2021b}, the authors studied platoon formation at multiple hubs when the decision-making of which platoons should form was distributed on coordinators at the hubs. In the above-mentioned work, the trucks are required to share their arrival schedules with the coordinators at the hubs.

 Transportation firms or individual truckers may need to keep their arrival times private for privacy or competitive reasons. For example, the arrival times may be sensitive information for a transportation firm or a trucker to share, as a competitor may use this information to take the firm's customers by offering earlier deliveries.  Thus, the arrival schedules may not always be available to the platoon coordinator. Platoon formation at a single hub when the trucks' arrivals are unknown to the hub coordinator in advance was studied in \cite{Johansson2020,Adler2020}, where the truck arrival process was modeled by a sequence of independent and identically distributed
random variables. This paper also considers a platoon formation problem with unknown arrivals but for multiple connected hubs instead of a single hub, as in  \cite{Johansson2020, Adler2020}. Our paper advances the state-of-the-art by developing efficient decentralized platoon formation policies, where the coordinator at each hub takes into account the downstream hubs' policies when computing its platoon release policy. Thus, the platoon formation problem in this paper becomes a multi-agent decentralized control problem.   We finally refer the reader to  \cite{Bhoopalam2018, Lesch2021} for extensive reviews on platoon coordination strategies.

%BEFORE
%Certain freight companies (and
%individual truck owners) may prefer to keep the arrival schedules of their trucks to the hubs private. Thus, the arrival schedules may not be always available to the platoon coordinator. Platoon formation at a single hub when the trucks' arrivals are unknown to the hub coordinator in advance was studied in \cite{Johansson2020,Adler2020}, where the truck arrival process was modeled by a sequence of independent and identically distributed
%random variables.   We finally refer the reviewer to  \cite{Bhoopalam2018, Lesch2021} for extensive reviews on platoon coordination strategies.

\subsection{Contributions}

% add?
%The hub-based platoon formation problem we consider in this paper is illustrated in Fig. \ref{Fig:Line_hubs}. The thick line in Fig. \ref{Fig:Line_hubs_a} represents the stretch of highway where a set of hubs are located and is referred to as the hub-corridor, and the grey lines represent roads that connect to the stretch of highway and trucks can join or leave the hub-corridor using them.  A coordinator at each hub decides which trucks will form platoons, as illustrated in Fig. \ref{Fig:Line_hubs_b}. The objective is to determine the trucks' departure times at hubs to form platoons with each other, considering the trade-off between platooning benefits and waiting costs. 

In this paper, we study the platoon formation in a set of hubs, as shown in Fig. \ref{Fig:Line_hubs_a}. In our set-up, each hub is equipped with a platoon coordinator, as shown in Fig. \ref{Fig:Line_hubs_b}, which decides the departure time of trucks from the hubs. The trucks that leave a hub at the same time form a platoon. We first develop the optimal release policies in the decentralized case as shown in Fig. \ref{Fig:Infstruct1}, where each coordinator only has access to the history of observed arrivals at its hub, the coordinators do not share any information with each other, and the arrival schedules of trucks are not known \emph{a priori}. We next develop coordination strategies in the distributed and centralized cases. In the distributed case, each coordinator has access to the release decisions of its preceding hub, whereas in the centralized case, the coordinators share their release decisions with each other and the arrival schedules of trucks are known by the coordinators.

\begin{figure}
\centering
\subfloat[Decentralized coordination]{\label{Fig:Infstruct1}	\includegraphics[width=7.6cm]{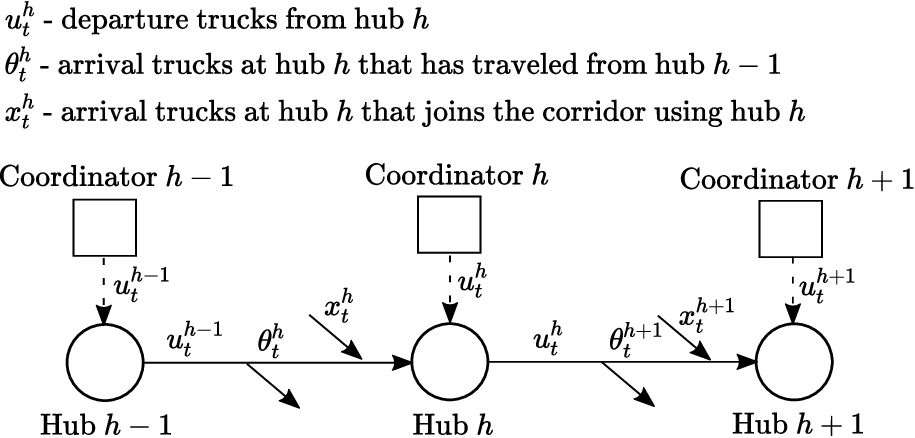}} \\

\subfloat[Distributed coordination]{\label{Fig:Infstruct2}	\includegraphics[width=7.6cm]{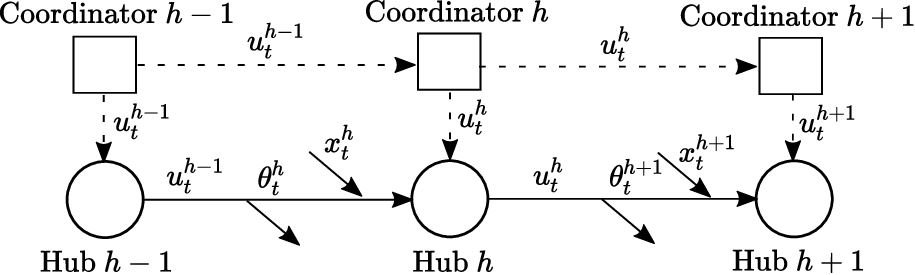}} \\

\subfloat[Centralized coordination]{\label{Fig:Infstruct3}	\includegraphics[width=7.0cm]{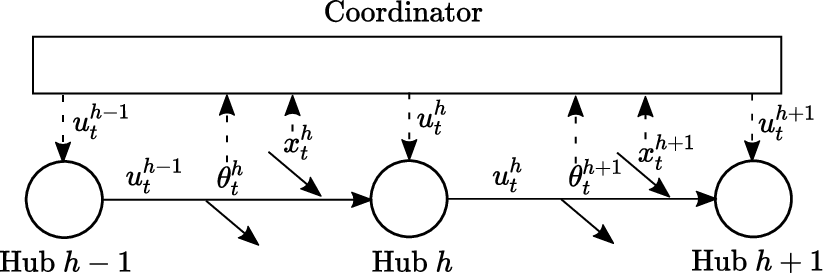}} 
 \caption{Coordination under different information structures}
 \label{Fig:Infstructs}
\end{figure}

The main contributions in this paper are as follows:

\begin{itemize}

\item   We investigate the structure of the optimal decentralized release policy at a hub when its truck arrivals, \emph{i.e.}, the number of trucks arriving at the hub,   are independent over time. Here, we show that a thresholding release policy is optimal. That is, it is optimal for the coordinator to release all the trucks located at the hub when the number of trucks is more than a threshold which is computed  by dynamic programming (DP). We use this result to develop approximate thresholding decentralized release policies in cases where truck arrivals are not independent over time by approximating the truck arrivals as independent.

\item We study the structure of the optimal decentralized release policy at a hub when the truck arrivals to its preceding hub are independent over time and the preceding hub follows an optimal release policy. We show that the optimal decentralized release policy, in this case, is in the form of a threshold policy which depends on a variable capturing the state of the preceding hub and can be computed using DP. This result is used to develop approximate thresholding decentralized release policies in cases where truck arrivals at the preceding hub are not independent over time by approximating the truck arrivals at the preceding hub as independent.

% THIS IS REMOVED R3.
%\item  We develop approximate thresholding decentralized release policies in the cases that the truck arrivals are not independent over time. The approximate decentralized release policies are computed by DP.

\item We propose solutions for the distributed and centralized cases, which are receding horizon solutions computed by DP.

\item We perform a simulation study for hub-based platoon formation along a highway in Sweden, where the decentralized release policies are compared with the  distributed and centralized release policies. The simulation study shows that the performance gap between the decentralized release policies and the distributed and centralized release policies is less than $8\%$. This observation indicates that the decentralized policies are potential solutions for hub-based platooning since they do not require communication between hubs and  achieve a similar performance compared with distributed and centralized policies.

\end{itemize}

%We note that this paper differs significantly from our previous work in \cite{Johansson2021} by that here we focus on decentralized coordination without information sharing among coordinators and trucks. This leads to more uncertainty and is realistic if transportation companies are unwilling to share information with others widely.

\subsection{Outline}

The outline of this paper is as follows. In Section \ref{Sec:system}, the system model of the hub-corridor is formulated, including the arrival process and the platoon release problem at each hub. In Section~\ref{Sec:firsthub}, the decentralized platoon release problem at a hub is studied under an independent-arrival assumption. In Section~ \ref{Sec:secondhub}, we also study the decentralized platoon release problem at a hub, but under the assumption that the preceding hub has independent arrivals and follows the developed release policy in Section~\ref{Sec:firsthub}. In Section \ref{Sec:Centsol}, release policies for distributed and centralized cases are developed. In Section \ref{Sec:simulation}, the proposed release policies are evaluated in a simulation study over a hub-corridor along a Swedish highway. Finally, the paper is concluded in Section \ref{Sec:conclusions}.

\section{System model}\label{Sec:system}

We consider a set of hubs in a transportation network, as shown in Fig. \ref{Fig:Line_hubs_a}, where trucks can wait for other trucks and form platoons. Examples of hubs include drivers resting areas, gas stations, and harbors. In our set-up, each hub has a platoon coordinator that decides whether the available trucks at the hub will depart as a platoon or wait for more trucks to join. We refer to the set of hubs located along a highway as a hub corridor, which is represented by the thick line in Fig.~\ref{Fig:Line_hubs_a}. The grey lines in Fig.~\ref{Fig:Line_hubs_a} represent roads used by trucks to join or leave the hub-corridor. The hubs are enumerated as~$1$ to $H$.

Consider the section of the hub-corridor that includes hubs~$h-1$ and $h$, as shown in Fig. \ref{Fig:Line_hubs_b}. The coordinator at hub~$h$ decides how many trucks to release from the hub at each time step. The released trucks depart from the hub in the form of a platoon. Let $u_t^h\in \mathbb{Z}_{\geq 0}$ denote the number of trucks released by the coordinator at hub $h$ at time step $t$. The state of a hub is defined as the number of trucks located at the hub. We use $n^h_t \in \mathbb{Z}_{\geq 0}$ to denote the state of hub $h$ at time step~$t$. The state of each hub changes dynamically based on the truck arrivals and the release decisions, as indicated in~Fig. \ref{Fig:Line_hubs_b}.  The state of hub $h$ has the following dynamics
\begin{align}\label{Eq:dynamics}
n^{h}_{t+1} & =n^{h}_{t}-u^{h}_t+x^{h}_{t+1}+\internalarrival{h}{t+1},
\end{align}
where $\internalarrival{h}{t}\in~\mathbb{Z}_{\geq 0}$ is the number of trucks that were released as a platoon at hub $h-1$ and arrive at hub $h$ at time step $t$, and $x^{h}_{t}\in \mathbb{Z}_{\geq 0}$ is the number of trucks that join the hub-corridor via hub $h$ at time step $t$. The state variable $n^{h}_{t}$
is a realization of the random variable
denoted $N^{h}_{t}$.

\begin{remark}
 Note that the trucks that arrive at the hub $h$ at time step $t$ can be divided into two groups. The first group is the trucks released as a platoon at hub $h-1$ and arrive at hub~ $h$. We denote the number of trucks in this group by $\theta^h_{t}$. The second group of trucks is those that enter the hub-corridor via hub $h$. The number of trucks in this group is denoted by $x^h_t$. For example, in Fig.~\ref{Fig:Line_hubs_a}, the trucks labeled with $A$ and~$B$ join the hub-corridor via the second hub, and the trucks labeled with $C$ and~$D$ which arrive at the second hub were released by the first hub as a platoon. In practice, the coordinator at hub $h$ can distinguish between these two groups of trucks by requiring the second group of trucks to send an acknowledgment message to the hub coordinator when they join the corridor. Thus, if the coordinator at hub $h$ does not receive a message from a truck, it knows that the truck was released from hub $h-1$.
\end{remark}

%%

%%%%
We assume that the truck arrivals are unknown to the coordinators \emph{a priori}. This is realistic if carriers or trucks keep their routes and schedules private or when the travel times are uncertain. To capture the uncertainty in the number of trucks arriving at hubs, the arrival variables $x^{h}_{t}$ and $\internalarrival{h}{t}$ are assumed to be realizations of random variables denoted $X^{h}_{t}$ and $\Internalarrival{h}{t}$, respectively. Throughout the paper, we assume that the random variables $X^{h}_{0}, X^{h}_{1},...$ are independent as these trucks have not been participating in a platoon formation process at a preceding hub, but the independent arrival assumption will in general not hold for the truck arrivals at hub $h$ that have been released at hub $h-1$.  The dependency between the hubs $h-1$ and $h$ captures that a truck that is released at the hub $h-1$ may exit the highway before reaching the hub $h$. We use $l^h$ to denote the likelihood of any truck exiting the highway between the hubs. Moreover, recall that the state variable $n_t^h$ is a realization of the random variable $N_t^h$, and the release decision $u_t^h$ is a mapping from $n_t^h$. Thus, the release decision $u_t^h$ can also be viewed as a random variable denoted by $U_t^h$.

The platoon release problem in the hub-corridor is defined as
\begin{equation}\label{Eq_problem_new}
\begin{aligned}
&\underset{\substack{\{\mu_0^h,\dots,\mu_T^h\}_h } }{\text{max}}&& \text{E}\Big [\sum_{h=1}^H \sum_{t=0}^T  R^h\big(N_t^h,U_t^{h}\big) \Big],\\
\end{aligned}
\end{equation}
where $T$ is the time horizon, and $R^h(n_t^h,u_t^h)$ denotes the reward of the coordinator of hub $h$ for releasing $u_t^h$ trucks when its state is~$n_t^h$. This reward can include the total cost-saving due to a platoon of length $u_t^h$ as well as the waiting cost of $n_t^h-u_t^h$ trucks at the hub. The release decision $u_t^h$ is computed using the release policy $\mu_t^h$ based on the available information to the coordinator at hub $h$, at time step $t$.

%%%%

In this paper, we assume that the reward function $R^h(n_t^h,u_t^h)$ is convex in $u_t^h$ and $n_t^h$. An important example of the reward function is the following convex piecewise-linear form:
	\begin{equation}\label{Eq:reward}
R^h(n_t^h,u_t^h)= \max\Big(0, \platooningbenefit{h}(u_t^h-1)\Big)- c(n_t^h-u_t^h),
\end{equation}
where the first term is the platooning profit for releasing $u_t^h$ trucks that form a platoon with $u_t^h-1$ follower trucks and $\platooningbenefit{h}$ is the profit per follower truck. The second term is the cost for having $n_t^h-u_t^h$ trucks waiting at the hub, and $c$ is the waiting cost per truck.

This paper focuses on decentralized solutions to the platoon release problem in \eqref{Eq_problem_new}, where the release decisions at each hub only depend on its local state and information. Computing the optimal decentralized release policies in \eqref{Eq_problem_new} becomes prohibitively difficult as the number of hubs increases.   The difficulty in computing an optimal decentralized solution to the problem \eqref{Eq_problem_new} arises because of the nested information structure of the hubs; each hub has to use its history of observed arrivals to estimate the history of all other hubs' observed arrivals and history of release decisions, and note, the variations of possible histories of observed arrivals and release decisions are numerous even for a few hubs. We refer the reader to \cite{Ho1972, Dave2019} for works on decentralized optimal control under nested information structures.  To solve this problem, in this paper, we decompose the optimization problem in \eqref{Eq_problem_new} into~$H$ sub-problems, one for each hub. In the rest of the paper, we will explore different decentralized solutions to the platoon release problem based on our decomposition approach.

\begin{remark}
 During off-peak hours when trucks are unaffected by congestion, the travel times have low variability, and assuming that travel times are deterministic and known a priori is justified. However, travel times may be uncertain during peak hours, and the travel times between hubs can then be modeled as stochastic. Our set-up can be used to study platoon formation under both deterministic and stochastic travel-time scenarios.  
\end{remark}

\begin{remark} The reward function in the convex piecewise-linear form in \eqref{Eq:reward} is accurate if the incremental platooning profit is equal for each follower truck in a platoon and deterministic and constant over the day. Moreover, the cost of waiting at the hub is linear and the same for all trucks. These assumptions may not hold due, for example, to the influence of factors such as traffic conditions and heterogeneous truck properties. However, our reward function in \eqref{Eq:reward} can be applied to a case with heterogeneous trucks where trucks have different platooning benefits. One way is to use the average platooning benefit of the trucks in our problem formulation. The reward function can also be extended to capture the impact of traffic by letting the platooning benefit be a time-varying (deterministic) signal where the benefit is low during peak hours and low during off-peak hours.  

\end{remark}

\section{Decentralized release policy: \\single-hub approach}\label{Sec:firsthub}

In this section, we propose a solution for the decentralized platoon release problem in \eqref{Eq_problem_new} by decomposing it into $H$ decoupled sub-problems, where each hub maximizes its own reward from platooning and irrespective of the decision-making behavior of other hubs. The decentralized platoon release problem at hub $h$ is 
\begin{equation}\label{Eq_problem}
\begin{aligned}
&\underset{\mu_0^h,\dots,\mu_T^h}{\text{max}}&& \text{E}\Big [\sum_{t=0}^T  R^h\big(N_t^h,U_t^{h}\big) \Big],\\
\end{aligned}
\end{equation}
where the objective is to maximize the reward of hub $h$ from platooning.  In the next subsection, we study the structure of the optimal release policy at a hub, under the single-hub approach, when its truck arrivals are independent over time. Later, we will use this result to derive an approximate solution for the single-hub approach when truck arrivals are dependent over time.

%The solution we propose in this section ignores that the arrival process of each hub is affected by other hubs' releasing behavior and is, therefore,  called the single-hub approach.

\subsection{Independent arrival case}

In this sub-section, we study the structure of the optimal release policy at hub $h$ under the single-hub approach and its  arrivals are independent over time, that is,  the random variables  $X_{0}^h,\dots,X_{T}^h,\Internalarrival{h}{0},\dots,\Internalarrival{h}{T}$ are assumed to be independent. Fig.~\ref{Ill3} illustrates the single-hub approach.  To study the structure of the optimal release policy at hub $h$,  we first derive the Bellman optimality equation associated with the optimal release policy of hub $h$. Using the Bellman's principle of optimality, the optimal value function associated with the state $n_t^h$ can be expressed as
\begin{align*}
&V^h_t(n_t^h)=\underset{ u_t^h \in \mathcal U^h_t(n_t^h)}{\text{max}}  R^h(n_t^h,u_t^h) + \\ & \hspace{3.5cm} \text{E} [V^h_{t+1}(n_t^h-u_t^h+X_{t+1}^h+\Internalarrival{h}{t+1}) ],
\end{align*}
where the decision variable $u_t^h$ can  take  values in the set $ \mathcal U^h_t(n^{h}_t)= \{0,\dots,n^{h}_t\}$ as the maximum number of trucks to release is the number of trucks at the hub. The value function at the terminal time step is  $V^h_T(n_T^h)~=~R^h(n_T^h,n_T^h)$, which is the reward for releasing the remaining trucks. The optimal value function can be computed using standard dynamic programming (DP), and the optimal decision can then be computed using the optimal value function. We omit to present the standard DP algorithm but refer the reader to \cite{Bertsekas2012} for a detailed presentation of the standard DP algorithm.

%The solution of the Bellman equation is the optimal release policy denoted $\mu_t^{h,*}$, for  $t=0,\dots,T$, which maps the state $n_t^h$ to an optimal release decision $u_t^{h,*}$.

%add caption

\begin{figure}
    \centering
    \includegraphics[width=5.5cm]{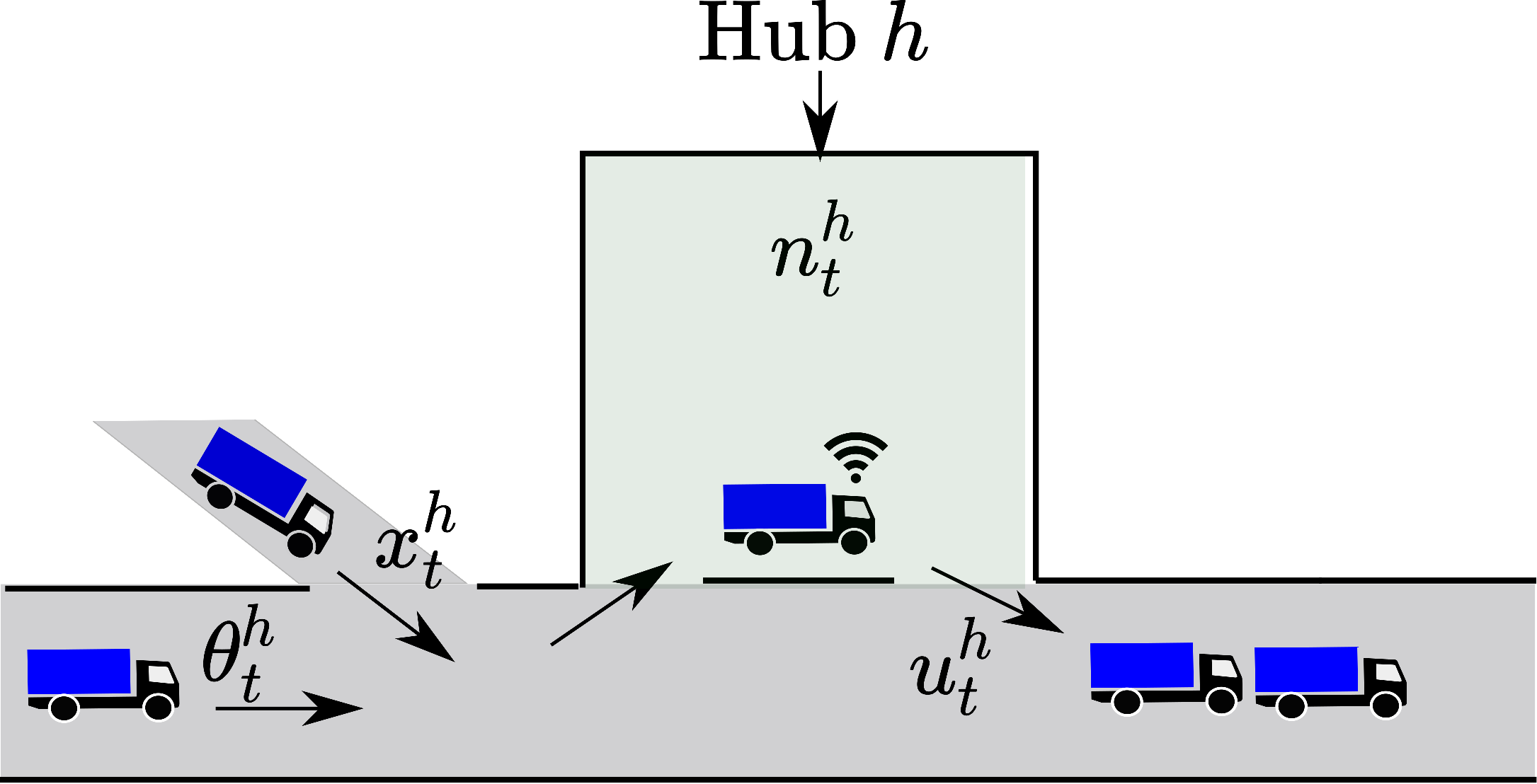}
    \caption{Single-hub model}
    \label{Ill3}
\end{figure}

The following theorem gives fundamental results of the structure of the optimal release policy under the single-hub approach when the truck arrivals are independent. The structural results significantly reduce the complexity of computing the optimal policy.

\begin{theorem}\label{Theorem:first_hub1}
Consider the decentralized platoon release problem in \eqref{Eq_problem}  when the truck arrivals are independent. The following holds:
\begin{itemize}
	\item The value function $V_t^h(n_t^h)$ is convex in $n_t^h$.
	\item The optimal release decision $u_t^{h,*}$ takes values in the set $\{0,n_t^h\}$. 
\end{itemize}
Furthermore,  if the reward function has the piecewise-linear form in \eqref{Eq:reward}, then there exists a threshold	$\threshold{h}{t}\in \mathbb{Z}_{\geq 0}$ such that 
	\begin{equation*}
	u_t^{h,*}=\begin{cases}
	n_t^h &\text{if } n_t^h\geq \threshold{h}{t} \\
	0 &\text{if } n_t^h< \threshold{h}{t}. \\
	\end{cases}
	\end{equation*}
The threshold $\threshold{h}{t}$ can be computed using dynamic programming.	
\end{theorem}

\begin{proof}
	See Appendix \ref{Appendix1}.
\end{proof}
According to Theorem \ref{Theorem:first_hub1}, the optimal decision of the coordinator of the hub $h$ is to either to release all its truck or release no trucks. This property is due to the convexity of the reward function, and significantly reduces the search space of optimal release policies, thus leading to a  light computational load when solving the Bellman equation by DP. Moreover, a release policy with the structure in Theorem~\ref{Theorem:first_hub1} can be used as a low complexity approximate release policy when the truck arrivals are
dependent. In such cases, the dependent arrivals can be approximated with an independent arrival process. Then, an approximate release policy can be computed using the structure of Theorem~\ref{Theorem:first_hub1} and the approximated arrival process.

\subsection{Dependent arrival case}

The independent arrival assumption in the previous sub-section is valid if hub $h$ is the only hub in the hub-corridor or if all trucks join the corridor using hub $h$. The arrivals at hub $h$ will be dependent if, for example, the arriving trucks have participated in a platoon formation process at hub $h-1$. In order to use the solution in Theorem \ref{Theorem:first_hub1} for the dependent arrival case, we approximate the joint distribution of the truck arrival process with the joint distribution of an independent arrival process.  That is, we approximate the random variables $\Internalarrival{h}{0},\dots, \Internalarrival{h}{T}$ as independent. An approximate release policy can be computed, as in the previous subsection, using the approximate arrival process with independent arrivals.

Algorithm \ref{alg} shows the procedure for computing the single-hub approximate release policy,  where for each hub $h$, one at a time, we first compute the empirical distribution of its arrivals based on arrival data during several episodes (each episode  spanning over the time steps $0$ to $T$). To this end,  we divide the time steps $0$ to $T$ into  $N$ intervals denoted $\mathcal T_0,\dots, \mathcal T_N$. Then, we approximate the distribution of $\Internalarrival{h}{t}$, for $t\in \mathcal T_j$, using the following equation
\begin{equation}\label{Eq:approx}
\Pr(\Internalarrival{h}{t}=\internalarrival{}{})= \frac{\sum \limits_{s\in \mathcal T_j} \sum\limits_{e \in \mathcal E} \boldsymbol{1}_{\internalarrivalsample{h}{s}{e}=\internalarrival{}{}}   }{\sum \limits_{s\in \mathcal T_j} \sum\limits_{e \in \mathcal E} 1},
\end{equation}
where the set of episodes is denoted $\mathcal E=\{1,\dots,E\}$  and $E$ is the number of episodes, $\internalarrivalsample{h}{s}{e}$ is the number of trucks in episode $e$ that have traveled from hub $h-1$ and arrive at hub $h$ at time step $s$, and $\boldsymbol{1}_{\internalarrivalsample{h}{s}{e}=\internalarrival{}{}}$ is the indicator function such that $\boldsymbol{1}_{\internalarrivalsample{h}{s}{e}=\internalarrival{}{}}=0$ if $\internalarrivalsample{h}{s}{e}\neq \internalarrival{}{}$ and $\boldsymbol{1}_{\internalarrivalsample{h}{s}{e}=\internalarrival{}{}}=1$ if $\internalarrivalsample{h}{s}{e}=\internalarrival{}{}$. We then use the empirical distributions of  $\Internalarrival{h}{0}$, ...,  $\Internalarrival{h}{T}$ to compute an approximate release policy for hub $h$. Note that the empirical estimator can be applied to either simulated or real arrival data and recall that the distributions of the random variables $X_0^h,...,X_T^h$ are independent as these trucks have not participated in a platoon formation process at a preceding hub.

\begin{algorithm} 
	\SetKwInOut{Input}{input}\SetKwInOut{Output}{output}
	\Input{Distribution of  $\{X_t^h\}_{t,h}$ }
	\Output{Approximate release policies $\{\mu_t^{h}\}_{t,h}$}
	\BlankLine 
	Generate $\{x^1_{t,e}\}_{t,e}$ using the distribution of $\{X^1_{t}\}_t$\\
	Compute optimal release policy  $\{\mu_t^{1}\}_t$ using DP and Theorem~\ref{Theorem:first_hub1}\\
	Compute the releases $\{u_{t,e}^1\}_{t,e}$ using $\{\mu_t^{1}\}_t$    and $\{x^1_{t,e}\}_{t,e}$  for each episode $e$
	\BlankLine

	\For{$h=2,\dots,H $ }{ 
	Generate $\{x^h_{t,e}\}_{t,e}$ using the distribution of $\{X^h_{t}\}_t$\\
	Generate arrivals $\{\internalarrivalsample{h}{t}{e}\}_{t,e}$ using $\{u_{t,e}^{h-1}\}_{t,e}$ for each episode $e$\\
	Approximate the distribution of $\{\Internalarrival{h}{t}\}_{t}$ using  \eqref{Eq:approx}\\
	Compute approximate  release policy  $\{\mu_t^{h}\}_t$ using DP and Theorem \ref{Theorem:first_hub1} \\
Compute the releases $\{u_{t,e}^h\}_{t,e}$ using $\{\mu_t^{h}\}_t$,  $\{\internalarrivalsample{h}{t}{e}\}_{t,e}$ and $\{x^h_{t,e}\}_{t,e}$  for each episode $e$ 
	}

	\caption{Computation of the single-hub approximate release policies}
	\label{alg}
\end{algorithm}

\section{Decentralized release policy: \\two-hub approach}\label{Sec:secondhub}

In this section, we propose a two-hub approach for the decentralized platoon release problem in \eqref{Eq_problem_new} where the  release problem is decomposed into $H$ decoupled sub-problems.  Under  the two-hub approach, each hub finds its optimal release policy based on the release policy of its preceding hub. In the following subsections, we first study the structure of the optimal release policy of a hub when its preceding hub has independent arrivals and the preceding hub follows the release policy in Theorem~\ref{Theorem:first_hub1}. We will then  derive an approximate release policy for a hub when the arrivals to its preceding hub are not independent in time.

\subsection{Independent arrival case}

%Before giving the Bellman optimality equation used to compute the optimal release policy, we first discuss the estimator state at hub~$h$.

In this sub-section, we study the structure of the optimal release policy at hub~$h$ when the arrivals to hub~$h-1$ are independent in time and hub $h-1$ follows the optimal release policy in Theorem~\ref{Theorem:first_hub1}. The two-hub model is illustrated in  Fig.~\ref{Ill4}. The random variables $\Internalarrival{h-1}{0}, \dots, \Internalarrival{h-1}{T}, X_{0}^{h-1}, \dots, X_{T}^{h-1}$ are the independent arrivals at hub $h-1$. The arrivals of trucks at hub $h$ from hub $h-1$, \emph{i.e.,} denoted by the random variables $\Internalarrival{h}{0}, \dots, \Internalarrival{h}{T}$,  are dependent in time since hub $h-1$ follows the  optimal release policy in Theorem~\ref{Theorem:first_hub1}. The random variables $X_0^{h},\dots,X_T^{h}$ are independent as the trucks that join the hub-corridor using hub~$h$ have not participated in a platoon formation process at a preceding hub.

\begin{figure}
    \centering
    \includegraphics[width=9cm]{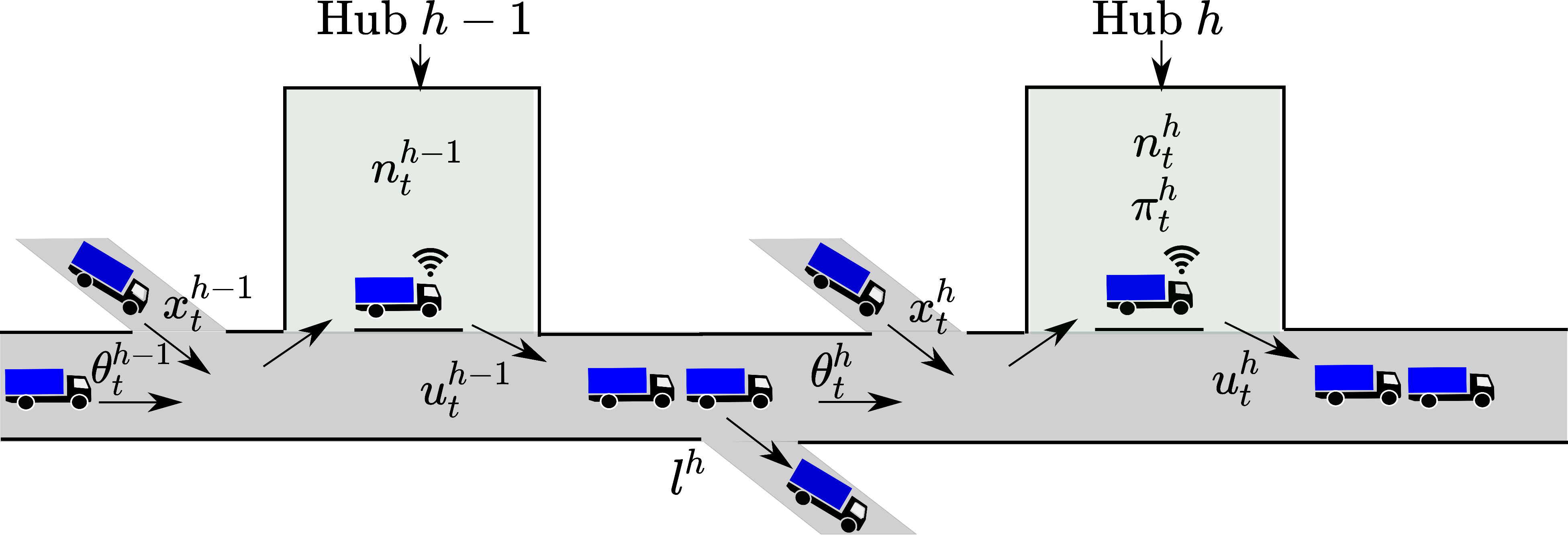}
    \caption{Two-hub model}
    \label{Ill4}
\end{figure}

 To study the structure of the optimal release policy at hub $h$, we derive the Bellman optimality equation.  Using the Bellman's principle of optimality,  the optimal value function can be expressed as
\begin{align}\label{Eq:value2}
&V^h_t(n_t^h,\filter{h}{t})=  \\ \nonumber & \hspace{0.5cm} \underset{ u_t^h \in \mathcal U^h_t(n_t^h)}{\text{max}} \ R^h(n_t^h,u_t^h) + \\ \nonumber & \hspace{1.2cm} \text{E} [V^h_{t+1}(n_t^h-u_t^h+X_{t+1}^h+\Internalarrival{h}{t+1},\filter{h}{t+1}],
\end{align}
where $\filter{h}{t}~=~\{\filterj{h}{t}{j}\}_{j\in \mathbb{Z}_{\geq 0}}$, and  $\filterj{h}{t}{i}$ is defined as
\begin{equation}
  \filterj{h}{t}{i}=\Pr(D_{t-\traveltime{h}}^{h-1}=i|\Internalarrival{h}{0}=~\internalarrival{h}{0}, \dots, \Internalarrival{h}{t}=~\internalarrival{h}{t}),  
\end{equation}
where $D_{t-\traveltime{h}}^{h-1}=N_{t-\traveltime{h}}^{h-1}-U_{t-\traveltime{h}}^{h-1}$ is the number of trucks at hub $h-1$ at time step $t-\traveltime{h}$ after releasing, and $\traveltime{h}$ is the travel time between the hubs $h-1$ and $h$. Thus, $\filterj{h}{t}{i}$ is the likelihood that $i$ trucks are at hub $h-1$ after releasing at time step $t-\traveltime{h}$ given the observed arrivals at hub $h$. The variable $\filter{h}{t+1}=\{\filterj{h}{t+1}{j}\}_{j\in \mathbb{Z}_{\geq 0}}$  can be computed recursively as
\begin{equation}\label{Eq:filter}
\filterj{h}{t+1}{j}=\frac{\sum_{i\geq 0} \filterj{h}{t}{i} \transprob{h}{t}{i}{j} \obsprob{h}{t}{i}{j}{\internalarrival{h}{t+1}} }{\sum_{i,j \geq 0} \filterj{h}{t}{i} \transprob{h}{t}{i}{j} \obsprob{h}{t}{i}{j}{\internalarrival{h}{t+1}} },
\end{equation} 
where $\transprob{h}{t}{i}{j}=\Pr(D_{t-\traveltime{h}+1}^{h-1}=j|D_{t-\traveltime{h}}^{h-1}=i)$ and $\obsprob{h}{t}{i}{j}{\internalarrival{h}{t+1}}~=~\Pr(\Internalarrival{h}{t+1}=\internalarrival{h}{t+1}|D_{t-\traveltime{h}+1}^{h-1}=j,D_{t-\traveltime{h}}^{h-1}=i)$. The probabilities $\transprob{h}{t}{i}{j}$ and $\obsprob{h}{t}{i}{j}{\internalarrival{h}{t+1}}$ are determined by the truck arrival distributions, the likelihood of leaving the corridor between the hubs, and the release policy of hub $h-1$, and their precise forms  are given in Appendix \ref{app_estimator}.   The optimal value function at the terminal time step is  $V^h_T(n_T^h,\filter{h}{T})~=~R^h(n_T^h,n_T^h)$.

%The control process for hubs $h-1$ and $h$ is illustrated in Fig.~\ref{Fig:two_hub_block}.
 
 %\begin{figure}[t]
%	\centering
%	\includegraphics[width=8cm]{Figures/Twohubblocknew4vert2}
%	\caption{Control process for hubs $h-1$ and $h$, including their dynamics, controllers, and the estimator at hub $h$.}
%	\label{Fig:two_hub_block}
%	\end{figure}

%and the optimal value function at the terminal time step is  $V^h_T(n_T^h,\filter{h}{T})~=~R^h(n_T^h,n_T^h)$.

The Bellman optimality equation in the form of \eqref{Eq:value2} is difficult to solve by DP since  $\filter{h}{t}$ is a continuous variable. In the rest of this subsection, we will derive the structure of the optimal value function and the optimal release policy that allows us to easily solve the Bellman optimality equation by DP. To this end,  let $w_t^h\in \mathbb{Z}_{\geq 0}$ denote the number of time steps since a non-zero truck arrival at hub $h$ from hub $h-1$, which can be computed recursively as
\begin{equation*}
w_{t+1}^h= \begin{cases}
w_{t}^h+1 &\text{if }  \internalarrival{h}{t+1}= 0 \\   0&\text{if } \internalarrival{h}{t+1}\neq 0.
\end{cases}
\end{equation*}

\begin{theorem}\label{Theorem:second_hub2}

	Consider the decentralized platoon release problem  at hub $h$ when the truck arrivals at hub $h-1$ are independent over time and its coordinator follows the release policy in Theorem \ref{Theorem:first_hub1}. Then, the following statements  hold:
\begin{itemize}
	\item   The optimal value function only depends on $n_t^h$ and $w_t^h$.
\item The value function $V_t^h(n_t^h, w_t^h)$ is convex in $n_t^h$.
	\item  The optimal release decision $ u_t^{h,*}$ takes  values in the set $\{0,n_t^h\}$.

\end{itemize}
Furthermore, if the reward function has the piecewise-linear form in \eqref{Eq:reward}, then there exists a threshold $\threshold{h}{t}(w_t^h)$, such that 
\begin{equation*}
u_t^{h,*}=\begin{cases}
n_t^h &\text{if } n_t^h\geq \threshold{h}{t}(w_t) \\
0 &\text{if } n_t^h< \threshold{h}{t}(w_t). \\
\end{cases}
\end{equation*}
The threshold  $\threshold{h}{t}(w_t^h)$ can be computed using dynamic programming.
\end{theorem}
\begin{proof}
	See Appendix \ref{Appendix4}. 
\end{proof}

Theorem \ref{Theorem:second_hub2} implies that the optimal release policy of hub~$h$ in the two-hub approach can be computed by solving the Bellman optimality equation using standard DP using $w_t^h$ instead of $\filter{h}{t}$. Note that $\filter{h}{t}$ is a continuous random variable which significantly complicates the computation of the optimal policy using DP. However, based on Theorem \ref{Theorem:second_hub2}, the optimal release policy can be computed by using $w_t^h$ which is a discrete variable. This significantly facilitates the computation of the optimal policy. Theorem \ref{Theorem:second_hub2} also indicates that similar  structural results continue to hold for the optimal release policy as those in Theorem~\ref{Theorem:first_hub1} when the arrivals at hub $h-1$ are independent and the coordinator at hub $h-1$ follows a release policy with the structure in  Theorem~\ref{Theorem:first_hub1}. This result is useful as it limits the candidate optimal release policies when seeking an optimal release policy by DP.

\subsection{Dependent arrival case}

The two-hub solution in the previous  subsection  is optimal under the assumptions that the arrivals at the preceding hub $h-1$ are independent and its coordinator follows the release policy in Theorem~\ref{Theorem:first_hub1}.  The independent-arrival assumption at hub $h-1$ is valid if hub $h-1$ is the first hub in the corridor. In this subsection, we propose an approximate two-hub solution for hub $h$ by approximating the joint distribution  of the arrival process at hub $h-1$ using independent arrivals. That is, we approximate the random variables $\Internalarrival{h-1}{0},\dots, \Internalarrival{h-1}{T}$ as independent. The approximate distribution allows us to   use Theorem~\ref{Theorem:second_hub2} to obtain an approximate release solution for the dependent case. The approximate two-hub release policy at hub $h$ is computed using the approximate arrival distribution at hub $h-1$ and DP, as explained in the previous subsection. 

Algorithm \ref{alg2} shows the procedure for computing the two-hub approximate release policy,  where for each hub $h$, we first compute the empirical distribution of the arrivals at hub $h-1$ using equation \eqref{Eq:approx}. The empirical distribution of the arrivals is then used as an approximation of the arrival distribution at hub $h-1$.   The empirical distribution of $\Internalarrival{h-1}{0},\dots, \Internalarrival{h-1}{T}$ is used to compute an approximate release policy for hub $h$.

\begin{algorithm} 
	\SetKwInOut{Input}{input}\SetKwInOut{Output}{output}
	\Input{Distribution of  $\{X_t^h\}_{t,h}$ }
	\Output{Approximate release policies $\{\mu^h_t\}_{t,h}$}
	\BlankLine
	Generate $\{x^1_{t,e}\}_{t,e}$ using the distribution of $\{X^1_{t}\}_t$\\
	Compute optimal release policy $\{\mu^1_t\}_t$ using DP and Theorem~\ref{Theorem:first_hub1} \\
   Generate releases $\{u_{t,e}^1\}_{t,e}$ using  $\{\mu^1_t\}_t$  and $\{x^1_{t,e}\}_{t,e}$ for each episode $e$ \\
   Generate $\{x^2_{t,e}\}_{t,e}$ using the distribution of $\{X^2_{t}\}_t$\\
	Compute optimal release policy  $\{\mu^2_t\}_t$ using DP and Theorem~\ref{Theorem:second_hub2}

	\BlankLine

	\For{$h=3,\dots,H $ }{ 
	Generate $\{x^h_{t,e}\}_{t,e}$ using the distribution of $\{X^h_{t}\}_t$\\	
Generate  arrivals $\{\internalarrivalsample{h-1}{t}{e}\}_{t,e}$ using $\{u_{t,e}^{h-2}\}_{t,e}$\\
	Approximate distribution of $\{\Internalarrival{h-1}{t}\}_{t}$ using  \eqref{Eq:approx}\\
	Compute approximate release policy  $\{\mu^h_t\}_t$ using DP and
Theorem 2\\
Generate  releases $\{u_{t,e}^{h-1}\}_{t,e}$ using $\{\mu^{h-1}_t\}_t$, $\{\internalarrivalsample{h-1}{t}{e}\}_{t,e}$ and $\{x^{h-1}_{t,e}\}_{t,e}$  for each episode $e$}

	\caption{Computation of the two-hub approximate release policies }
	\label{alg2}
\end{algorithm}

\iffalse

\begin{algorithm} 
	\SetKwInOut{Input}{input}\SetKwInOut{Output}{output}
	\Input{Distribution of  $\{X_t^h\}_{t,h}$ }
	\Output{Release policies $\{\mu^h_t\}_{t,h}$}
	\BlankLine
	Compute release policy $\{\mu^1_t\}_t$ by DP\\
    Simulate releases $\{u_{t,e}^1\}_{t,e}$\\
	Compute release policy  $\{\mu^2_t\}_t$ by DP

	\BlankLine

	\For{$h=3,\dots,H $ }{ 
	Simulate arrivals $\{\internalarrivalsample{h-1}{t}{e}\}_{t,e}$ using $\{u_{t,e}^{h-2}\}_{t,e}$\\
	Approximate distribution of $\{\Internalarrival{h-1}{t}\}_{t}$ using  \eqref{Eq:approx}\\
	Compute release policy  $\{\mu^h_t\}_t$ by DP\\
	Simulate releases $\{u_{t,e}^{h-1}\}_{t,e}$ 
	}

	\caption{Computation of the two-hub approximate release policies }
	\label{alg2}
\end{algorithm}	

\fi
\section{Distributed and centralized release policies}\label{Sec:Centsol}

In this section, we propose distributed and centralized release policies, which will be used as benchmarks in Section~ \ref{Sec:simulation} to evaluate the performance of the proposed decentralized release policies and investigate the value of communication among hubs. Under the distributed release policy, each coordinator informs its proceeding coordinator in the hub-corridor about its release decisions. Under the centralized release policy,  coordinators share their release decisions and have \emph{a priori} knowledge of their arrivals.

\subsection{Distributed case}

In the distributed case, the coordinator at each hub receives the release decisions of its preceding hub and computes its release decisions using a receding horizon solution with $L$ time steps as the horizon. The distributed release policy requires communication links between hubs which will require installation and maintenance, whereas the decentralized release policies do not require communication among hubs. Let $\mathcal M^h_t= \{ u^{h-1}_{t-\traveltime{h}+1},\dots,u^{h-1}_{t-\traveltime{h}+L}\}$ denote the release decisions of hub $h-1$ which are known by hub $h$  at time step $t$. Then, the Bellman optimality equation for hub $h$ can be expressed as 
\begin{align*}
&V^h_s(n_s^h,\mathcal M^h_t)=\underset{ u_s^h \in \mathcal U^h_s(n_s^h)}{\text{max}} \ R(n_s^h,u_s^h) + \\ & \hspace{2cm} \text{E} [V^h_{s+1}(n_s^h-u_s^h+X_{s+1}^h+\Internalarrival{h}{s+1},\mathcal M^h_t)|\mathcal M^h_t ],
\end{align*}
for $s=t,\dots,t+L-1$, and the value function at the terminal time step of the horizon is  $V^h_{t+L}(n_{t+L}^h) = R^h(n_{t+L}^h,n_{t+M}^h)$. The Bellman optimality equation is solved recursively by DP. Then, the computed release decision for time step $t$ is implemented, while the set of releases at the preceding hub and the computed release decisions for time steps $t+1,\dots,t+L$ are updated at the next step.

%The information in this set determines the probability distributions of the random variables  $\Internalarrival{h}{t+1},\dots, \Internalarrival{h}{t+L}$, as the travel time between the hubs is $\traveltime{h}$ and  each released truck at hub $h-1$ arrives at hub $h$ or leave the corridor between the hubs, with probability $l^h$. 

%The probability distribution of the random variables  $\Internalarrival{h}{t+1},\dots, \Internalarrival{h}{t+L}$ is determined by the  set of releases decisions  $\mathcal M^h_t$ as the trucks that are released at hub $h-1$ either arrives at hub $h$ or leaves the hub-corridor in between the hubs. Recall that the probability for any truck to leave the hub-corridor between hubs $h-1$ and $h$ is $l^h$. Thus, given $\mathcal M^h_t$, we have  $\Internalarrival{h}{s} \sim \text{Binominal}(u^{h-1}_{s-\traveltime{h}},1-l^h)$, for $ t+1< s\leq t+L$. 

\subsection{Centralized case}

In the centralized case, the coordinator at each hub computes its release decisions using a receding horizon solution where the arrivals within the horizon of $L$ time steps are known. The arrivals are known as the coordinators sharing their release decisions in the centralized case, and the trucks inform the coordinators about their arrivals beforehand. The centralized release policy requires knowledge of the schedule and the route of trucks and communication between hubs. However, the decentralized release policies do not rely on any information exchange between hubs. The set of information available to  the coordinator of hub $h$  at time step~$t$ is denoted as $\mathcal N^h_t= \{ x^{h}_{t+1},\dots,x^{h}_{t+L}, \internalarrival{h}{t+1} ,\dots,\internalarrival{h}{t+L} \}$. The Bellman optimality equation used by the coordinator at hub $h$ to compute the release decision at time step $t$ is expressed as
\begin{align*}
&V^h_s(n_s^h,\mathcal N^h_t)=\underset{ u_s^h \in \mathcal U^h_s(n_s^h)}{\text{max}} \ R(n_s^h,u_s^h) + \\ & \hspace{3.5cm} V^h_{s+1}(n_s^h-u_s^h+x_{s+1}^h+\internalarrival{h}{s+1},\mathcal N^h_t),
\end{align*}
for $s=t,\dots,t+L-1$, and the value function at the terminal time step of the horizon is  $V^h_{t+L}(n_{t+L}^h) = R^h(n_{t+L}^h,n_{t+L}^h)$. The Bellman optimality equation is solved recursively by DP, and similar to the distributed case,  the release decision for time step $t$ is implemented, while the set of information and the release decisions for time steps $t+1,\dots,t+L$ are updated at the next time step.

\section{Simulation study}\label{Sec:simulation}

In this section, we perform a simulation study over a hub-corridor with three hubs in northern Sweden.  We first explain the simulation setup, including the hub locations and arrival distributions. Then, we study the performance of the developed coordination solutions under deterministic travel times, which is accurate if the impact from surrounding traffic on the highway is ignored. We also study the impact of  uncertainty in travel times between the hubs on the performance of the coordination policies.

\subsection{Setup}

 Fig.~\ref{Fig:corr} shows a hub corridor between  Lule{\aa} and Sundsvall in northern Sweden. The hubs are located near the cities Lule{\aa}, Skellefte{\aa}, Ume{\aa}, and the hub-corridor ends near the city of Sundsvall, Sweden. The length of the road segments connecting the hubs are $131$~km, $136$~km, and $263$ km, and we assume that trucks travel with a speed of $80$ km/h under free-flow conditions.  To obtain realistic arrival distributions for the trucks that join the corridor, we use the real data shown in Fig. \ref{Fig:corr2} which were collected during two working days in 2018--2019 by the Swedish Transport Administration \cite{TV2022}. This figure shows the average hourly truck count on roads that connect to each hub location in Fig.~\ref{Fig:corr} (excluding the roads connecting Lule{\aa}, Skellefte{\aa}, Ume{\aa}, and Sundsvall).

\begin{figure}
	\centering
	\includegraphics[width=4.5cm]{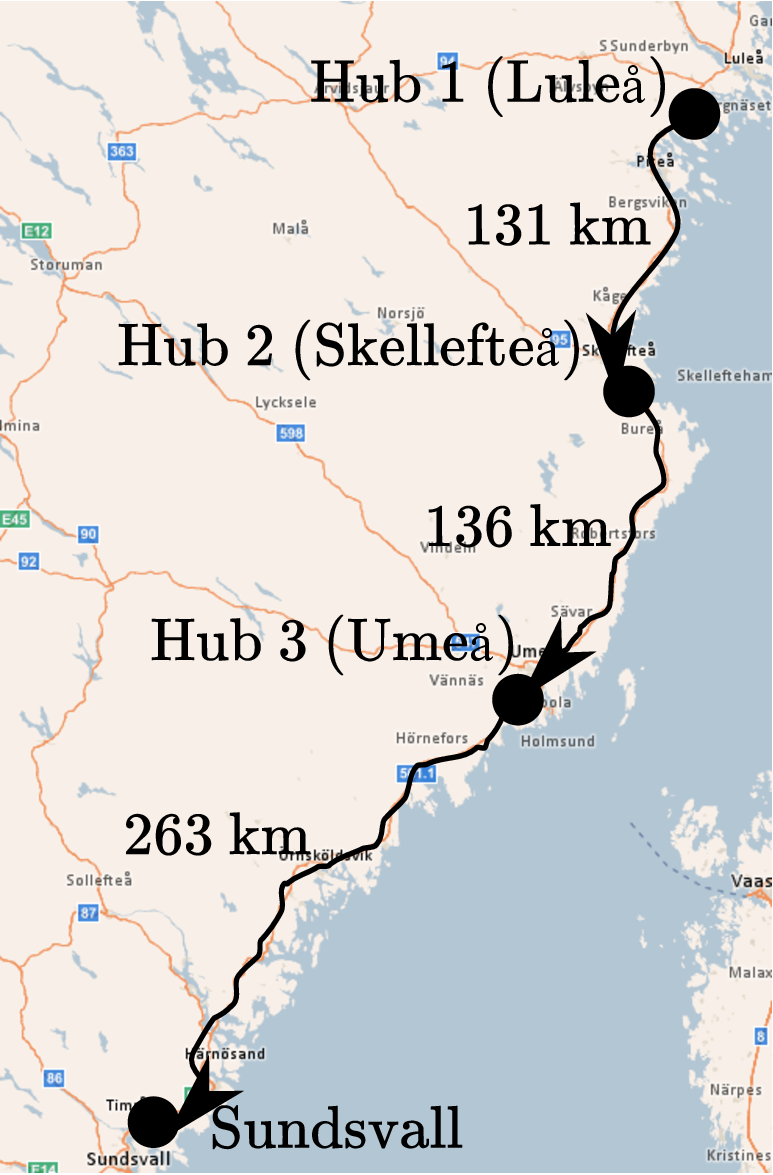}
	\caption{Three-hub-corridor along a highway in Sweden.}
	\label{Fig:corr}
	
\end{figure}

\begin{figure}
	\centering
	\includegraphics[width=8.3cm]{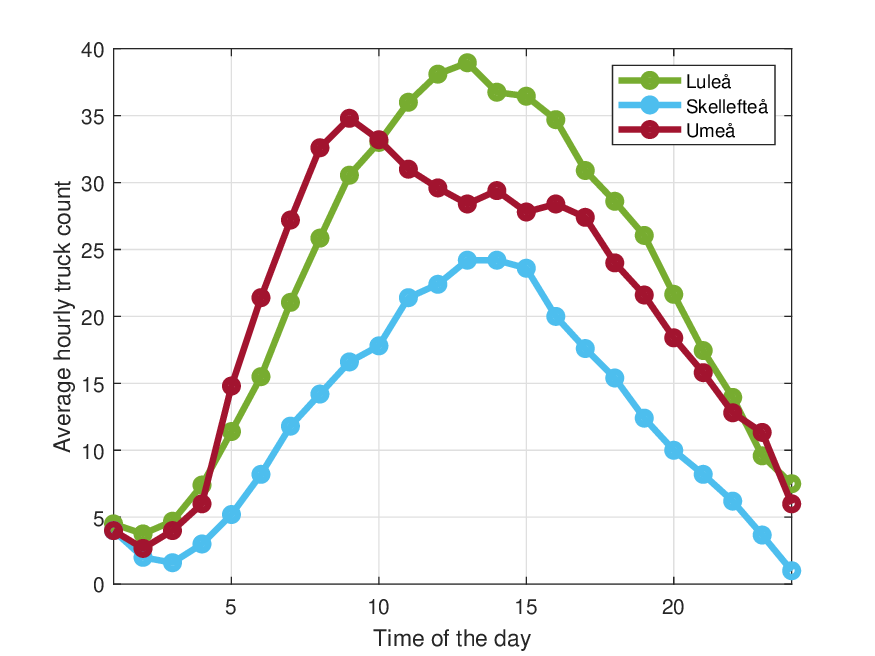}
	\caption{The average hourly truck count on roads connecting to Lule{\aa}, Skellefte{\aa}, and Ume{\aa} (excluding roads connecting Lule{\aa}, Skellefte{\aa}, Ume{\aa}, and Sundsvall). The data was collected by the Swedish Transport Administration \cite{TV2022}. }
	\label{Fig:corr2}
	
\end{figure}

In our simulations, we assume that the number of trucks joining the corridor using hub $h$ at time $t$, \emph{i.e.}, $X_t^h$ is Poisson distributed with a time-varying mean. More precisely,  $X^{h}_{t} \sim \text{Poisson}(\lambda^h_t)$, for $t=~0,\dots, T$, where the mean of the process, \emph{i.e.},  $\lambda^h_t$, is set according to the data in Fig. \ref{Fig:corr}.  The time step length is set to one minute  in the simulations.  We will study the performance of the proposed decentralized policies under both known travel-times and uncertain travel-times.

%will first assume known travel times between the hubs, then assume the travel times are truncated normally distributed random variables with positive values and the free-flow travel-times between the hubs as mean. The impact of travel time uncertainty will be studied by varying the standard deviation of the travel times.  

We use the reward function in \eqref{Eq:reward} for our simulations, where the platooning benefit ($b^h$) is $10 \%$ fuel consumption reduction for the follower trucks. The cost of fuel is assumed to be  $5$ SEK/km in our simulations. The waiting cost per truck in our simulations is equal the hourly driver cost,  which is around 200 SEK/h  \cite{Trafikanalys2016}, or approximately $c=3.33$ SEK per time step. The following results are generated by~$50$~Monte Carlo simulations.

%The decentralized single-hub release policy (Section \ref{Sec:firsthub}) and the decentralized two-hub release policy (Section \ref{Sec:secondhub}) will be compared with the distributed and centralized release policies (Section \ref{Sec:Centsol}).  The horizon length for the distributed and centralized release policies is $L=30$ time steps.

% Where to mention  $L=30$ time steps.

\subsection{Deterministic travel times scenario}

In this subsection, we study the performance of different platoon release policies in the deterministic travel time scenario where the travel times between hubs are known \emph{a priori} and set to the free-flow travel times. Fig. \ref{Fig:Profit} shows the (hourly) average reward in a period of 24 hours under different release policies. According to this figure, the decentralized two-hub solution generally results in a higher average hourly reward than the decentralized single-hub solution. This is because the former policy takes the arrival process and release behavior of its preceding hub into account. Fig. \ref{Fig:Profit} also shows that the distributed and centralized release policies outperform the decentralized release policies by a small margin. This indicates that the performance loss due to the decentralized structure of the one-hub and two-hub solutions is relatively small.

\begin{figure}
\centering
\subfloat[First hub]{\label{Fig:Prof1}	\includegraphics[width=9cm]{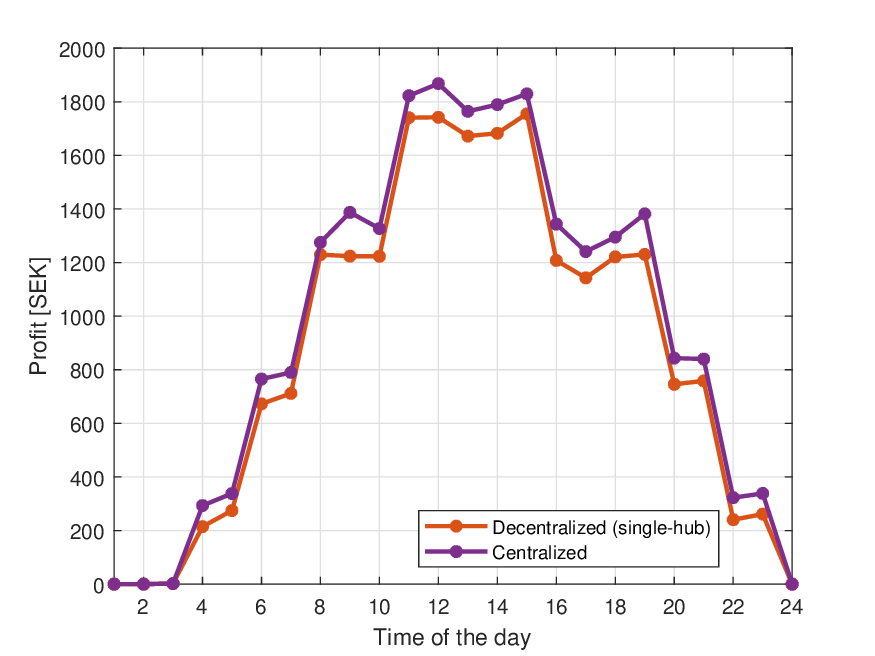}} \\

\subfloat[Second hub]{\label{Fig:Prof2} \includegraphics[width=9cm]{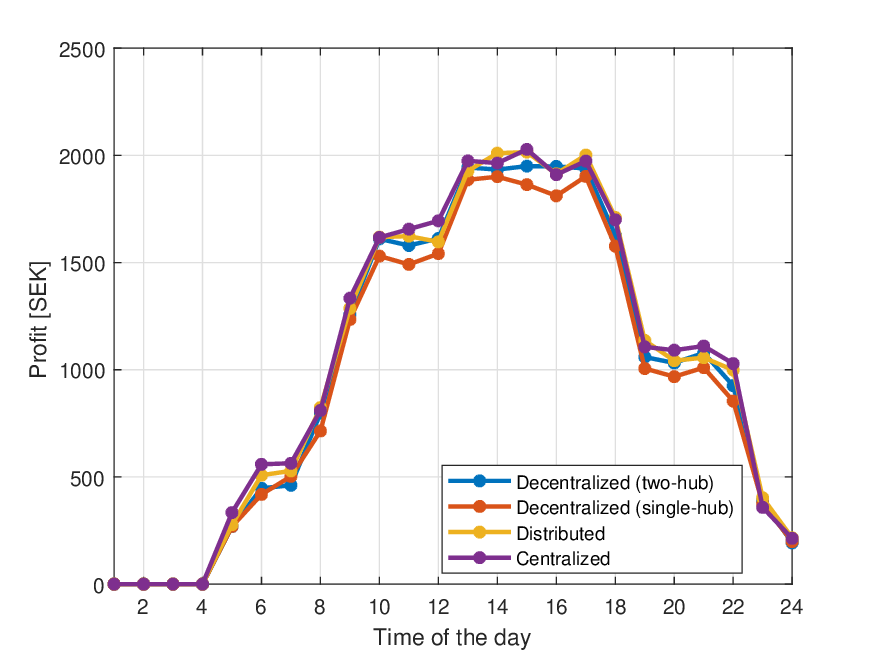}} \\

\subfloat[Third hub]{\label{Fig:Prof3}
	\includegraphics[width=9cm]{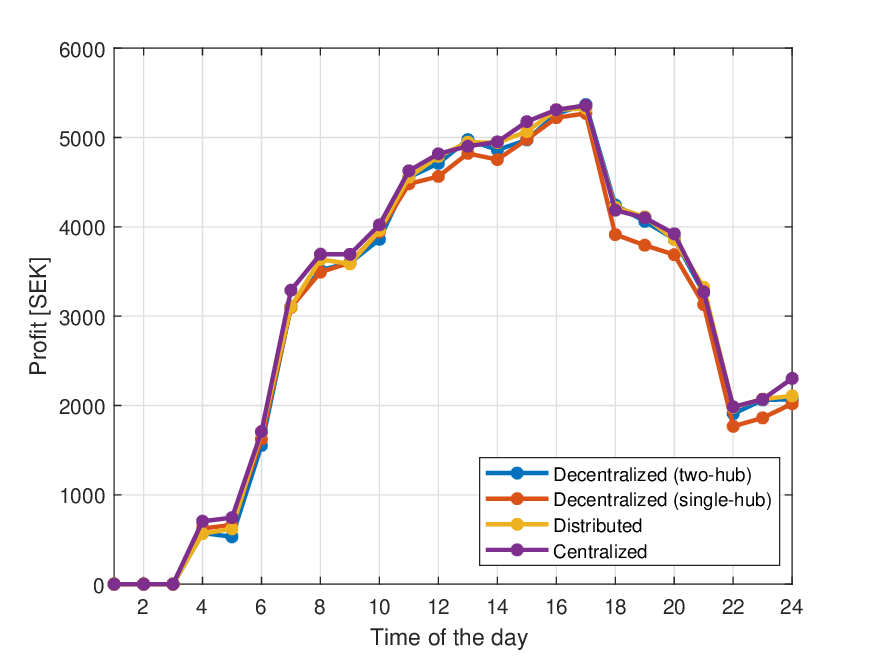}}
 \caption{Profits as a function of time. In the simulations,  the likelihood of any released truck to leave the corridor before the next hub is set to $l^h=0.5$.}
 \label{Fig:Profit}
\end{figure}

%% Platoon length

Fig. \ref{Fig:PZ} shows the (hourly) average number of released trucks and its standard deviation for each hub. The number of released trucks is referred to as release size in this figure. According to this figure, more trucks are released on average under the decentralized two-hub policy than the decentralized single-hub policy. The standard deviation of the release sizes is also higher for the decentralized two-hub policy. The increased variability in release sizes is due to that the threshold for the decentralized two-hub policy depends on the variable $w_t^h$, \emph{i.e.,} the number of time steps since a non-zero arrival was observed, as shown in Theorem \ref{Theorem:second_hub2}. Fig. \ref{Fig:PZ} also shows that the distributed and centralized policies generally have larger release sizes and higher standard deviations than the decentralized policies. This implies that the centralized and distributed policies are more efficient in forming platoons than the decentralized policies.

\begin{figure}
	\centering
	\subfloat[First hub]{\label{Fig:PZ1}	\includegraphics[width=9cm]{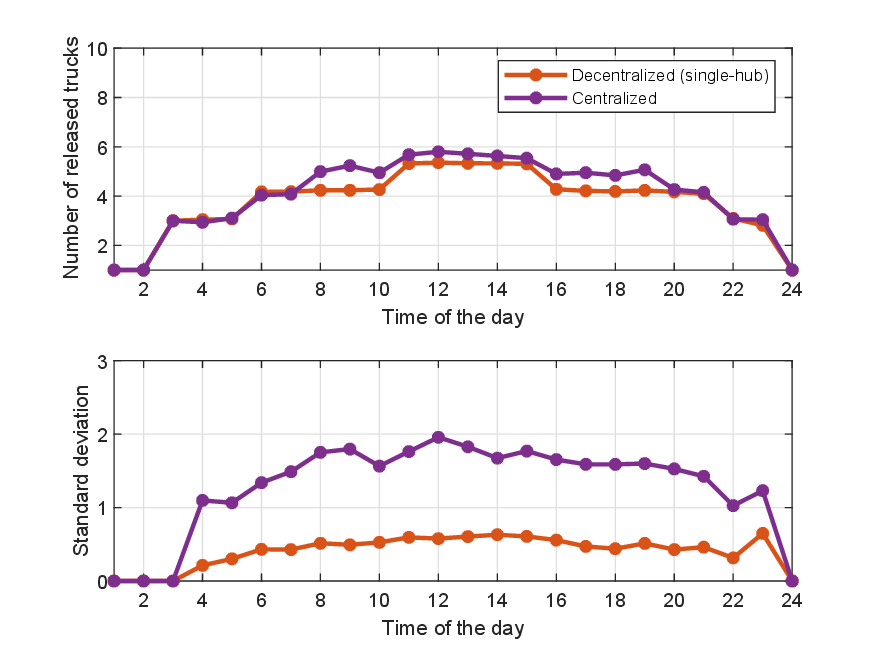}} \\
	
	\subfloat[Second hub]{\label{Fig:PZ2} \includegraphics[width=9cm]{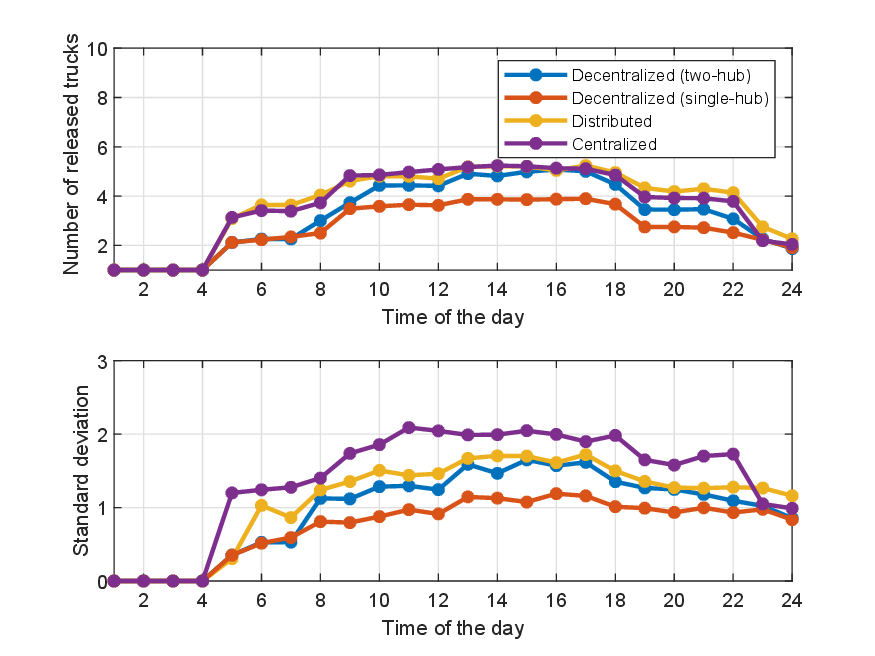}} \\
	
	\subfloat[Third hub]{\label{Fig:PZ3}
		\includegraphics[width=9cm]{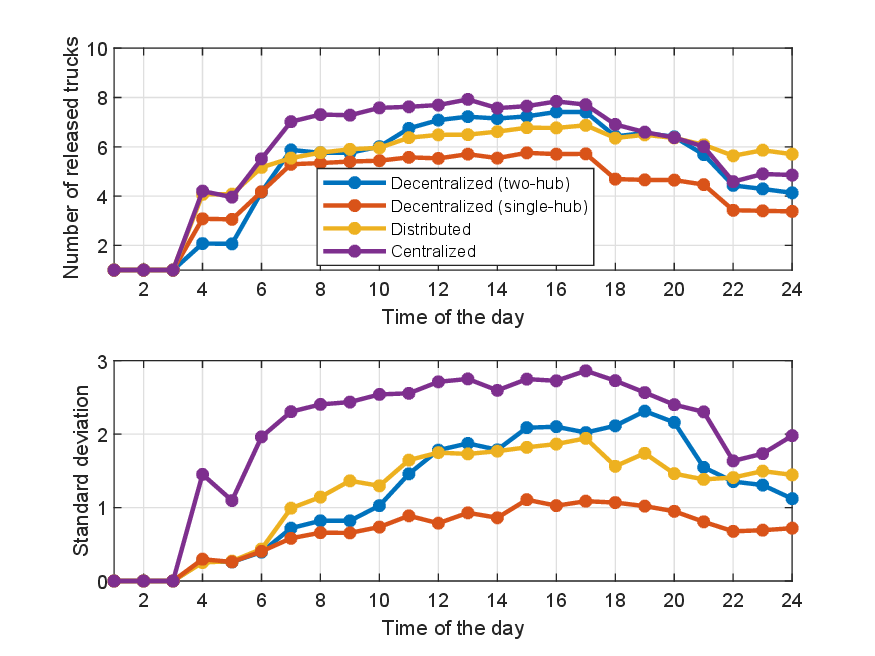}}
	\caption{Platoon length and its standard deviation as a function of time. In the simulations,  the likelihood of any released truck to leave the corridor before the next hub is set to $l^h=0.5$.}
	\label{Fig:PZ}
\end{figure}

Fig. \ref{Fig:SensP} shows the average daily reward generated at the hubs as a function of the variable $l^h$, \emph{i.e.}, the likelihood that a truck leaves the hub-corridor before reaching the next hub in the hub-corridor. Based on this figure, the daily reward decreases as the variable $l^h$ increases. This is because the trucks will visit fewer hubs on average as $l^h$ becomes large, reducing the number of trucks participating in platoons.   Moreover, Fig.~ \ref{Fig:SensP} shows that the difference between the decentralized and distributed policies is small when $l^h$ becomes large.  This figure also shows that the centralized release policy achieves higher rewards even when $l^h$ is high. The largest difference between the decentralized and centralized policies is $8\%$ and occurs at $l^h=0.8$.

%the probability of leaving the hub-corridor between hubs is high. The largest difference between the decentralized and distributed policies is $3.5\%$ and occurs at the leaving probability of $0.5$.

\subsection{Uncertain travel times scenario}

In this subsection, we will study the impact of travel-time uncertainty on the performance of different release policies. To this end, we model the travel times between the hubs as truncated Gaussian distributed random variables with the free-flow travel times between  hubs as the mean of the distribution. The standard deviation of the travel times is varied to study the impact of travel time uncertainty. Fig. \ref{Fig:SensTTun} shows the average daily reward of the hubs as a function of the standard deviation of the travel-times. According to Fig. \ref{Fig:SensTTun}, the reward slightly decreases with increased uncertainty in travel times. The performance loss of decentralized policies is less than $2\%$ when the standard deviation of the travel-time is $10$ minutes. Also,  the performance loss of the distributed and centralized policies is less than $5\%$  when the standard deviation of the travel time is $10$ minutes.

% add to capion he probability of leaving the corridor between hubs is set to $0.5$. The figure shows that the profits slightly decrease with increased uncertainty in travel times. 

\begin{figure}[t]
	\centering
	\includegraphics[width=9cm]{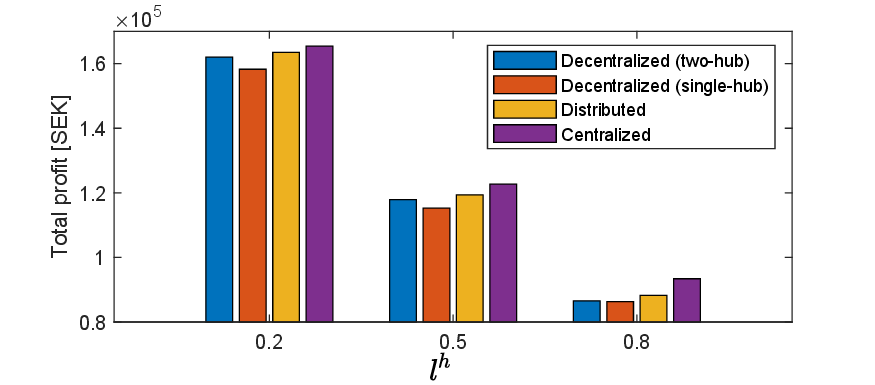}
	\caption{Total profit as a function of the likelihood of any released truck to leave the corridor before the next hub.}
	\label{Fig:SensP}
	
\end{figure}

\begin{figure}[t]
	\centering
	\includegraphics[width=9cm]{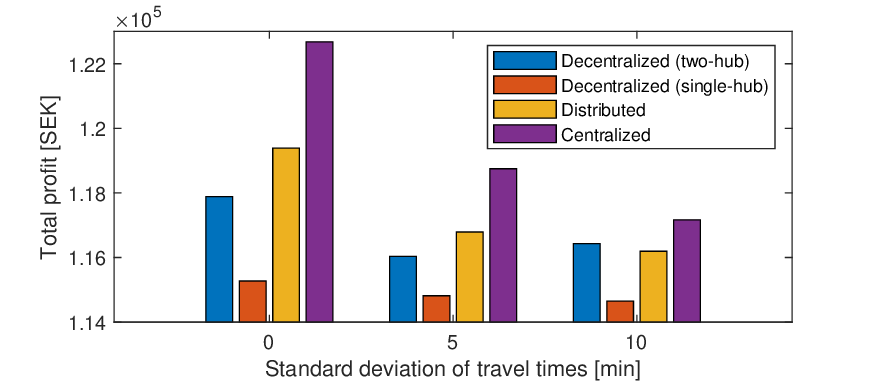}
	\caption{Total profit as a function of the  standard deviation of the travel times. In the simulations,  the likelihood of any released truck to leave the corridor before the next hub is set to $l^h=0.5$.}
	\label{Fig:SensTTun}
	
\end{figure}

\section{Conclusion and future work} \label{Sec:conclusions}

This paper studied the platoon release problem in a hub-corridor, in which release times at hubs are coordinated to form platoons. We focused on the decentralized case without communication between coordinators, where the release decisions are based on statistical information of truck arrivals. Under this information structure, we proposed two different decentralized release policies. The first decentralized release policy was shown to be optimal when its arrivals from the preceding hub are independent. The second decentralized release policy is shown to be optimal when the preceding hub has independent arrivals and follows an optimal release policy. In order to evaluate the value of communication among hubs, we also proposed distributed and centralized release policies.

A simulation study of a hub-corridor with three hubs in Sweden was performed to compare the proposed release policies and evaluate the value of sharing information. The simulation study showed that the decentralized release policies achieved profits close to that of the distributed and centralized. More precisely, the profits of the decentralized policy was less than $3.5\%$ and $8\%$ lower than of the distributed and centralized, respectively. This suggests that a viable platooning system can be obtained without sharing truck routes and schedules across the system, which may otherwise be one of the bottlenecks for platoon cooperation among different transportation companies. Moreover, the simulations showed that the likelihood of leaving the system between hubs heavily  affects the platoon  profits of the proposed release policies. For instance, when the likelihood  of leaving the corridor between hubs is high, the benefit of exploiting trucks from the preceding neighboring hub is limited.

For future research, we aim to study platoon release problems in networks with general typologies, expanding the line graph (hub-corridor) considered in this paper.  A possibility is to include platoon routing in the decision of the hub coordinators.  Another promising future research direction is to use reinforcement learning to seek approximate solutions when the truck arrival process is unknown. Furthermore, another avenue for future work is to investigate if similar properties shown for the platoon release problem in this paper hold when using a more realistic form of the reward function that, for example, captures the influence of traffic and heterogeneous truck properties.   Final possible future extensions are to use a traffic simulator, for example, SUMO, to model the travel time uncertainty in the simulation and to include hub and road capacity constraints in the problem formulation.

%Finally, we aim to use the models developed in this paper to seek approximate release polices using the toolbox provided in reinforcement learning. 

%simple->naive

\appendices
\section{Proof of Theorem \ref{Theorem:first_hub1}}\label{Appendix1}

Before providing the proof, we define
\begin{equation*}
G^h_t(n_t^h,u_t^h)=R^h(n_t^h,u_t^h) + \text{E} [V^h_{t+1}(n_t^h-u_t^h+X_{t+1}^h+\Internalarrival{h}{t+1})],
\end{equation*}
and note that 
\begin{equation}\label{Eq: proof_1}
V^h_t(n^h_t)= \underset{u_t^h\in \mathcal U_t^h(n_t^h) }{\text{max}} G^h_t(n_t^h,u_t^h).
\end{equation}
%We prove the first part of the theorem by induction.

First, assume that $V^h_{t+1}(n^h_{t+1})$ is convex in $n^h_{t+1}$. Then, since convexity is preserved under affine maps and positively weighted summations, we have that $\text{E} [V^h_{t+1}(n^h_t-u^h_t+X_{t+1}^h+\Internalarrival{h}{t+1})]$ is convex in both $n^h_{t}$ and $u^h_t$. Therefore, $G^h_t(n^h_t,u^h_t)$ is also convex in both $n^h_{t}$ and $u^h_t$. Consider then the relaxed version of the maximization problem in \eqref{Eq: proof_1} with the decision space $[0,n_t^h]$ instead of $\mathcal U_t^h(n_t^h)=\{0,\dots, n_t^h\}$. The relaxed maximization problem is of a convex function over a compact and convex set.  Therefore, the maximizer is at one of the boundary points of the relaxed decision space $[0,n_t^h]$, see  \cite{Rockafellar1970} for proof. Since the boundary points of $[0,n_t^h]$ are elements in  $\mathcal U_t^h(n_t^h)$, the maximizer of the maximization problem in \eqref{Eq: proof_1} is also one of these two boundary points. That is, the optimal release decision $ u_{t}^{h,*}$ takes values in in the set $\{0,n^h_{t}\}$. By this fact we have that
	\begin{align*}
V^h_t(n^h_t)= \underset{ }{\text{max}} \{ G^h_t(n^h_t,0), G_t(n^h_t,n^h_t)  \},
\end{align*}
which is convex in $n^h_t$ since  $G^h_t(n^h_t,0)$ and $G^h_t(n^h_t,n^h_t)$ are convex functions in $n^h_t$ and the maximizer of two convex functions is convex. Finally, $V^h_{T}(n^h_{T})=R^h(n^h_T,n^h_T)$ is convex by assumption. Thus,  $V^h_{t}(n^h_{t})$ is convex for all $t\leq T$ by induction. Conclusions in first part of Theorem~\ref{Theorem:first_hub1} follow.

In the second part of Theorem \ref{Theorem:first_hub1}, we assume the piecewise-linear form of $R^h(n_t^h,u_t^h)$ in \eqref{Eq:reward}. Then we have for $n^h_t\geq 1$ that 
\begin{equation*}
G^h_t(n^h_t,n^h_t)= n^h_t\platooningbenefit{h}+F_t,
\end{equation*}
where $F_t=-\platooningbenefit{h}+\text{E} [V^h_{t+1}(X_{t+1}^h+\Internalarrival{h}{t+1})]$  and 
\begin{equation*}
G^h_t(n^h_t,0)= -cn^h_t+\text{E} [V^h_{t+1}(n^h_t+X_{t+1}^h+\Internalarrival{h}{t+1})].
\end{equation*}
Since the optimal release decision $u_t^{h,*}$ takes values in the set $\{0,n^h_{t}\}$, we have that $u_t^{h,*}=n^h_{t}$ if and only if
\small
\begin{align}\label{Eq:theorem2_diff}
\Delta (n^h_t) &=G^h_t(n^h_t,n^h_t)-G^h_t(n^h_t,0) \\
& =n^h_t\platooningbenefit{h}+F+cn^h_t-  \text{E} [V^h_{t+1}(n^h_t+X_{t+1}^h+\Internalarrival{h}{t+1}) ]\geq 0. \nonumber
\end{align}
\normalsize
 To show that there exists a threshold $\threshold{h}{t}$ such that $u_t^{h,*}=n^h_{t}$ if and only if $n^h_{t}\geq \threshold{h}{t}$, we show that the optimal release decision is $u_t^{h,*}=n^h_{t}\geq 1$ (when the number of trucks at hub $h$ is $n^h_{t}$)  implies the optimal release decision $u_t^{h,*}=n^h_{t}+1$ (when the number of trucks at hub $h$ is $n^h_{t}+1$). To show this, it is sufficient to show that $\Delta (n^h_t)$ is non-decreasing in $n^h_t$. First, assume that there exists a threshold $\threshold{h}{t+1}$ such that $u_{t+1}^{h,*}=n^h_{t+1}$  if and only if $n^h_{t+1}\geq \threshold{h}{t+1}$. This implies that $V^h_{t+1}(n^h_{t+1}+1)-V^h_{t+1}(n^h_{t+1})=\platooningbenefit{h}$ if $n^h_{t+1}\geq \threshold{h}{t+1}$. Since convex functions have increasing differences it follows that
\begin{equation}\label{Eq:appendix_theorem2_1}
V^h_{t+1}(n^h_{t+1}+1)-V^h_{t+1}(n^h_{t+1})\leq   \platooningbenefit{h},
\end{equation} 
for all $n^h_{t+1}$.
Thus, equations \eqref{Eq:theorem2_diff} and \eqref{Eq:appendix_theorem2_1} imply
\begin{align*}
\Delta (n^h_t+1)-	\Delta (n^h_t)=   \\ & \hspace{-3.2cm}
\platooningbenefit{h}+c-\text{E} [V^h_{t+1}(n^h_t+1+X_{t+1}^h+\Internalarrival{h}{t+1})]-  \\ & \hspace{-2cm} +\text{E} [V^h_{t+1}(n^h_t+X_{t+1}^h+\Internalarrival{h}{t+1}) ] \geq  \platooningbenefit{h}+c-\platooningbenefit{h}>0.
\end{align*}
We regard the threshold $\threshold{h}{T}=1$ at the terminal time step and the conclusion in the second part of Theorem \ref{Theorem:first_hub1} follows by induction.

\section{Transition and observation probabilities}\label{app_estimator}

The optimal release policy at hub $h-1$ is in the form given in Theorem~\ref{Theorem:first_hub1}. Consequently, there exist a set $\releaseset{t-\traveltime{h}}{h-1}  \subseteq \mathbb{Z}_{\geq 0}$ such that the optimal release decision $u_{t-\traveltime{h}}^{h-1,*}=n_{t-\traveltime{h}}^{h-1}$ if  $n_{t-\traveltime{h}}^{h-1} \in \releaseset{t-\traveltime{h}}{h-1}$ and otherwise $u_{t-\traveltime{h}}^{h-1,*}=0$. Then, for $i\notin \releaseset{t-\traveltime{h}}{h-1}$ and $j\notin \releaseset{t-\traveltime{h}+1}{h-1}$, we have \small
\begin{align*}\label{explicitp}
  \transprob{h}{t}{i}{j}  \!=\! \!\begin{cases}
\Pr( X^{h-1}_{t-\traveltime{h}+1}+i=j) & \text{if } j\geq i \land j\neq 0\\
\Pr(X^{h-1}_{t-\traveltime{h}+1}+i \in  \releaseset{t-\traveltime{h}+1}{h-1}) & \text{if } i\neq 0 \land j= 0 \\
\Pr(X^{h-1}_{t-\traveltime{h}+1}\in  \releaseset{t-\traveltime{h}+1} {h-1}\lor  X^{h-1}_{t-\traveltime{h}+1} \!= \!0) \hspace{-0.25cm} & \text{if } i=0 \land j=0 \\
0 & \text{else,}
\end{cases}
\end{align*}  \normalsize
where the first case is when the number of trucks at hub $h-1$ after the release decision is non-decreasing due to not releasing, the second case is when the number of trucks at hub $h-1$ after the release decision becomes zero due to releasing, and the third case is when the number of trucks at hub $h-1$ after the release decision remains zero due to releasing or no arrivals.  Furthermore,   for $\internalarrival{}{}\neq0$ and $j= 0$, we have
\small
\begin{align*}
\obsprob{h}{t}{i}{0}{\internalarrival{}{}} & = \\ & \hspace{-1cm}\sum \limits_{n\in \releaseset{t-\traveltime{h}+1}{h-1} }\!\!\!\!\!\!\!\! \Pr( X^{h-1}_{t-\traveltime{h}+1}+i=n)  \Pr(\Internalarrival{h}{t+1}=\internalarrival{}{}| U^{h-1}_{t-\traveltime{h}+1 }=n  ) /\transprob{h}{t}{i}{0} ,
\end{align*}
\normalsize
where $\Pr(\Internalarrival{h}{t+1}=\internalarrival{}{}| U^{h-1}_{t-\traveltime{h}+1 }=n  )$ is determined by  $l^h$, and for  $\internalarrival{}{}\neq0$ and $j\neq  0$, we have $\obsprob{h}{t}{i}{j}{\internalarrival{}{}}=0$.

\section{Proof of Theorem \ref{Theorem:second_hub2}}\label{Appendix4}

 The optimal release decision at hub $h-1$, that is $u_t^{h-1,*}$, takes values in the set $\{0,n_t^{h-1}\}$, as shown in Theorem  \ref{Theorem:first_hub1}. Thus, $\obsprob{h}{t}{i}{j}{\internalarrival{h}{t+1}}=0$ for $j>0$ and $\internalarrival{h}{t+1}\neq 0$. This implies $\filterj{h}{t+1}{0}=1$ and $\filterj{h}{t+1}{j}=0$ for $j>0$ when $\internalarrival{h}{t+1}\neq 0$. That is,  $\filter{h}{t}$ is reset to $\filter{h}{t}=(1,0,\dots,0)$ when $\internalarrival{h}{t}\neq 0$ and the number of time steps since a non-zero arrival was observed (or more precisely $w^h_t$) is the only variable affecting $\filter{h}{t}$. We can therefore write the bellman equation as 
\begin{align*}
V_t(n_t^h,w_t^h)  = &  \underset{ u_t^h \in \mathcal U^h_t(n_t^h)}{\text{max}}  R^h(n_t^h,u_t^h) +  \\  & \hspace{-0.7cm} \text{E} [V^h_{t+1}(n^h_t-u_t^h+X_{t+1}^h+\Internalarrival{h}{t+1}, w^h_{t+1}) ].
\end{align*}

% left our how w affects the arrival probabilites
%or
%\begin{align*}
%V^h_t(n^h_t,w^h_t)=  &   \underset{ u^h_t \in \mathcal U^h_t(n_t^h)}{\text{max}}  R^h(n_t^h,u_t^h) +  \\  & \hspace{-0.5cm} \sum_{\internalarrival{}{} x} \Pr(\Internalarrival{h}{t+1}=\internalarrival{}{}|W_t^h=w_t^h) \Pr(X^h_{t+1}=x) \\[-10pt]   &  \hspace{1cm}
%V^h_{t+1}(n^h_t-u^h_t+x+\internalarrival{}{}, \zeta(w^h_t, \internalarrival{}{})), 
%\end{align*}
%where $\Pr(\Internalarrival{h}{t+1}=\internalarrival{}{}|W_t^h=w_t^h)= \sum_{ij} \eta^h_{t,i}(w^h_t)  \transprob{h}{t}{i}{j}\obsprob{h}{t}{i}{j}{\internalarrival{}{}}$. 

%%

The proof of the rest of Theorem \ref{Theorem:second_hub2} is similar to the proof of Theorem~ \ref{Theorem:first_hub1} and we therefore omit most details  but provide the step which is least trivial. If  $V^h_{t+1}(n_{n+1}^h,w^h_{t+1})$ is convex in $n_{n+1}^h$, then  $\text{E} [V^h_{t+1}(n^h_t-u_t^h+X_{t+1}^h+\Internalarrival{h}{t+1}, w^h_{t+1}) ]$ is convex in $n^h_t$ and $u_t^h$, and  by following similar steps as in Appendix~\ref{Appendix1}, this implies that $V^h_{t}(n_t^h,w^h_{t})$ is convex in $n_t^h$. Then, $V^h_{t}(n_t^h,w^h_{t})$ is convex for all $t\leq T$ by induction since  $V^h_{T}(n_T^h,w^h_{T})=R^h(n_t^h,n_t^h)$ is convex in $n_t^h$. The remaining part of the proof follows similar steps as the  proof of Theorem~\ref{Theorem:first_hub1} but including the state $w^h_t$ in the notations. 

%\begin{align*}
%\text{E} [V^h_{t+1}(n^h_t-u_t^h+X_{t+1}^h+\Internalarrival{h}{t+1}, \zeta(w^h_t, \Internalarrival{h}{t+1})) ] &= \\  & \hspace{-5cm}\sum_{\internalarrival{}{} x} \Pr(\Internalarrival{h}{t+1}=\internalarrival{}{}|W_t^h=w_t^h)\Pr(X^h_{t+1}=x)  \\[-10pt]   & \hspace{-3cm} 
%V^h_{t+1}(n^h_t-u^h_t+x+\internalarrival{}{}, \zeta(w^h_t, \internalarrival{}{}))
%\end{align*}

%\section*{Acknowledgment}

%The authors would like to thank...

% Can use something like this to put references on a page
% by themselves when using endfloat and the captionsoff option.
\ifCLASSOPTIONcaptionsoff
  \newpage
\fi

% trigger a \newpage just before the given reference
% number - used to balance the columns on the last page
% adjust value as needed - may need to be readjusted if
% the document is modified later
%\IEEEtriggeratref{8}
% The "triggered" command can be changed if desired:
%\IEEEtriggercmd{\enlargethispage{-5in}}

% references section

% can use a bibliography generated by BibTeX as a .bbl file
% BibTeX documentation can be easily obtained at:
% http://mirror.ctan.org/biblio/bibtex/contrib/doc/
% The IEEEtran BibTeX style support page is at:
% http://www.michaelshell.org/tex/ieeetran/bibtex/
%\bibliographystyle{IEEEtran}
% argument is your BibTeX string definitions and bibliography database(s)
%\bibliography{IEEEabrv,../bib/paper}
%
% <OR> manually copy in the resultant .bbl file
% set second argument of \begin to the number of references
% (used to reserve space for the reference number labels box)

% biography section

% You can push biographies down or up by placing
% a \vfill before or after them. The 

% Can be used to pull up biographies so that the bottom of the last one
% is flush with the other column.
%\enlargethispage{-5in}
\bibliographystyle{ieeetr}
\bibliography{RefDatabase}

\end{document}